\documentclass[11pt, a4paper]{article}
\usepackage{authblk}
\usepackage{comment}

\usepackage{setspace}
\usepackage{fullpage}

\usepackage{amssymb}
\usepackage{amsmath}
\usepackage{amsfonts}
\usepackage{amsthm}
\allowdisplaybreaks[1]
\usepackage{enumerate}
\usepackage{enumitem}
\usepackage{float}

\usepackage{indentfirst}

\usepackage{booktabs}
\usepackage{multirow}
\usepackage{threeparttable}
\usepackage{adjustbox}

\usepackage{footnotehyper} 
\makesavenoteenv{table}  

\usepackage{graphicx}
\usepackage[dvipsnames]{xcolor}

\usepackage{listings}
\usepackage[ruled,vlined]{algorithm2e}
\usepackage[backref=page,colorlinks,linkcolor=teal,anchorcolor=blue,citecolor=teal]{hyperref}
\usepackage[refpage]{nomencl}
\usepackage[title]{appendix}

\usepackage[english]{babel}
\usepackage[T1]{fontenc}
\usepackage[normalem]{ulem}

\usepackage{framed}
\usepackage[modulo]{lineno}
\usepackage{cite}

\usepackage{tikz}
\usetikzlibrary{quantikz2}

\usepackage{cleveref}
\crefname{ineq}{inequality}{inequalities}
\creflabelformat{ineq}{#2{\upshape(#1)}#3}
\crefname{fact}{fact}{facts}
\crefname{claim}{claim}{claims}
\crefname{equation}{equation}{equations}
\crefname{algorithm}{algorithm}{algorithm}
\crefname{remark}{remark}{remarks}
\crefname{conjecture}{conjecture}{conjectures}

\usepackage{thmtools}
\declaretheorem[style=plain,numberwithin=section]{theorem}
\declaretheorem[style=plain,numberlike=theorem]{lemma,proposition,corollary}
\declaretheorem[style=remark,numberlike=theorem]{remark}
\declaretheorem[style=plain,numberlike=theorem]{definition,conjecture,fact,claim}

\numberwithin{equation}{section}

\allowdisplaybreaks[4]

\newcommand{\BQP}{\textnormal{\textsf{BQP}}\xspace}

\newcommand{\PreciseBQP}{\textnormal{\textsf{PreciseBQP}}\xspace}

\newcommand{\SZK}{\textnormal{\textsf{SZK}}\xspace}
\newcommand{\QSZK}{\textnormal{\textsf{QSZK}}\xspace}
\newcommand{\NP}{\textnormal{\textsf{NP}}\xspace}

\newcommand{\QMA}{\textnormal{\textsf{QMA}}\xspace}
\newcommand{\PreciseQMA}{\textnormal{\textsf{PreciseQMA}}\xspace}
\newcommand{\QIPii}{\textnormal{\textsf{QIP(2)}}\xspace}
\newcommand{\QIP}{\textnormal{\textsf{QIP}}\xspace}
\newcommand{\PP}{\textnormal{\textsf{PP}}\xspace}

\newcommand{\AMcapcoAM}{\mathsf{AM \cap coAM}}
\newcommand{\StoqMA}{\textnormal{\textsf{StoqMA}}\xspace}

\newcommand{\NQP}{\textnormal{\textsf{NQP}}\xspace}
\newcommand{\coCeP}{\mathsf{coC_{=}P}}
\newcommand{\Class}{\textnormal{\textsf{C}}\xspace}

\renewcommand{\bra}[1]{\langle{#1}|}
\renewcommand{\ket}[1]{|{#1}\rangle}

\newcommand{\ketbra}[2]{\left| #1 \right\rangle \left\langle #2 \right|}
\newcommand{\innerprod}[2]{\left\langle #1 | #2 \right\rangle}

\newcommand{\Tr}{\mathrm{Tr}}
\newcommand{\rank}{\mathrm{rank}}
\newcommand{\sign}{\mathrm{sgn}}
\newcommand{\SWAP}{\mathrm{SWAP}}

\newcommand{\td}{\mathrm{T}}
\newcommand{\binH}{\mathrm{H_2}}
\newcommand{\D}{\mathrm{D}}

\newcommand{\textoverline}[1]{$\overline{\mbox{#1}}$}

\newcommand{\SD}{\mathrm{SD}}
\newcommand{\TD}{\mathrm{TD}}
\newcommand{\Hsquare}{\mathrm{H}^2}
\renewcommand{\H}{\mathrm{H}}
\newcommand{\HSsquare}{\mathrm{HS}^2}
\newcommand{\HS}{\mathrm{HS}}
\newcommand{\Bsquare}{\mathrm{B}^2}
\newcommand{\B}{\mathrm{B}}
\newcommand{\Col}{\mathrm{Col}}
\newcommand{\F}{\mathrm{F}}
\newcommand{\JS}{\mathrm{JS}}
\renewcommand{\S}{\mathrm{S}}
\newcommand{\QJS}{\textnormal{\textrm{QJS}}\xspace}
\newcommand{\measQJS}{\mathrm{QJS}^{\rm meas}}

\newcommand{\QHD}{\mathrm{QH^2}}
\newcommand{\QTD}{\textnormal{\textrm{QTD}}\xspace}
\newcommand{\measQTD}{\mathrm{QTD}^{\rm meas}}
\newcommand{\SDP}{\textnormal{\textsc{SDP}}\xspace}
\newcommand{\coSDP}{\textnormal{\textoverline{\textsc{SDP}}}\xspace}
\newcommand{\TDP}{\textnormal{\textsc{TDP}}\xspace}
\newcommand{\JSP}{\textnormal{\textsc{JSP}}\xspace}
\newcommand{\QJSP}{\textnormal{\textsc{QJSP}}\xspace}
\newcommand{\QEDP}{\textnormal{\textsc{QEDP}}\xspace}
\newcommand{\QTDP}{\textnormal{\textsc{QTDP}}\xspace}
\newcommand{\QSDP}{\textnormal{\textsc{QSDP}}\xspace}
\newcommand{\coQSDP}{\texorpdfstring{\textnormal{\textoverline{\textsc{QSDP}}}}\xspace}
\newcommand{\measQTDP}{\textnormal{\textsc{measQTDP}}\xspace}

\newcommand{\prob}{\textnormal{\textsc{Prob}}\xspace}
\newcommand{\GapQSD}{\textnormal{\textsc{GapQSD}}\xspace}

\renewcommand{\Pr}[1]{\mathrm{Pr}\left[#1\right]}

\newcommand{\binset}{\{0,1\}}
\newcommand{\supp}[1]{\mathrm{supp}\left( #1 \right)}
\newcommand{\poly}{\mathrm{poly}}
\newcommand{\diag}{\mathrm{diag}}

\newcommand{\trho}{\tilde{\rho}}
\newcommand{\bbC}{\mathbb{C}}
\newcommand{\bbE}{\mathbb{E}}

\newcommand{\bbR}{\mathbb{R}}
\newcommand{\bbN}{\mathbb{N}}

\newcommand{\calA}{\mathcal{A}}
\newcommand{\calB}{\mathcal{B}}
\newcommand{\calE}{\mathcal{E}}
\newcommand{\calH}{\mathcal{H}}

\newcommand{\calM}{\mathcal{M}}

\newcommand{\calU}{\mathcal{U}}
\newcommand{\calX}{\mathcal{X}}

\newcommand{\sfA}{\mathsf{A}}
\newcommand{\sfB}{\mathsf{B}}
\newcommand{\sfC}{\mathsf{C}}
\newcommand{\sfF}{\mathsf{F}}
\newcommand{\sfS}{\mathsf{S}}
\newcommand{\sfY}{\mathsf{Y}}
\newcommand{\sfZ}{\mathsf{Z}}

\DeclarePairedDelimiter\rbra{\lparen}{\rparen}
\DeclarePairedDelimiter\sbra{\lbrack}{\rbrack}
\DeclarePairedDelimiter\cbra{\{}{\}}
\DeclarePairedDelimiter\abs{\lvert}{\rvert}
\DeclarePairedDelimiter\norm{\lVert}{\rVert}
\DeclarePairedDelimiter\ceil{\lceil}{\rceil}

\begin{document}

% Reduce the line spacing between equation and text
\setlength{\abovedisplayskip}{6pt}
\setlength{\belowdisplayskip}{6pt}

\title{Quantum state testing beyond the polarizing regime and \\quantum triangular discrimination}
\author{Yupan Liu\thanks{Email: yupan.liu.e6@math.nagoya-u.ac.jp}}
\affil{Graduate School of Mathematics, Nagoya University}
\date{}
\maketitle

\pagenumbering{roman}
\thispagestyle{empty}

\begin{abstract}
The complexity class Quantum Statistical Zero-Knowledge (\QSZK{}) captures computational difficulties of the time-bounded quantum state testing problem with respect to the trace distance, deciding whether $\td(\rho_0,\rho_1)$ is at least $\alpha$ or at most $\beta$, known as the Quantum State Distinguishability Problem (\QSDP{}) introduced by \hyperlink{cite.Wat02}{Watrous (FOCS 2002)}. 
However, $\QSDP[\alpha,\beta]$ is in \QSZK{} only within the constant polarizing regime, where $\alpha$ and $\beta$ are constants satisfying $\alpha^2 > \beta$ (rather than $\alpha > \beta$), similar to its classical counterpart shown by \hyperlink{cite.SV97}{Sahai and Vadhan (JACM 2003)} due to the polarization lemma (error reduction for \SDP{}). 

Recently, \hyperlink{cite.BDRV19}{Berman, Degwekar, Rothblum, and Vasudevan (TCC 2019)} extended the \SZK{} containment of \SDP{} beyond the polarizing regime via the time-bounded distribution testing problems with respect to the triangular discrimination and the Jensen-Shannon divergence.  
Our work introduces \textit{proper} quantum analogs for these problems by defining quantum counterparts for triangular discrimination. We investigate whether the quantum analogs behave similarly to their classical counterparts and examine the limitations of existing approaches to polarization regarding quantum distances. These new \QSZK{}-complete problems improve \QSZK{} containments of \QSDP{} beyond the polarizing regime and establish a simple \QSZK{}-hardness for the quantum entropy difference problem (\QEDP{}) defined by \hyperlink{cite.BASTS10}{Ben-Aroya, Schwartz, and Ta-Shma (ToC 2010)}. 
Furthermore, we prove that \QSDP{} with some exponentially small errors is in \PP{}, while the same problem without error is in \NQP{}. 
\end{abstract}

\newpage
\tableofcontents
\thispagestyle{empty}
\newpage
\pagenumbering{arabic}

%%%%%%%%%%%%%%%%%%%%%%%%%%%%%%%%%%%%%%%%%%%%%%%%%%%%%

\section{Introduction}

The quantum state testing problem is generally about telling whether one quantum (mixed) state is close to the other with respect to a chosen distance-like measure, which is an instance of the emerging field of quantum property testing~\cite{MdW16}. This problem generalizes the classical question of testing whether two probability distributions are close, known as distribution testing~\cite{Canonne20}. 
This problem typically focuses on the number of samples (sample complexity) needed to distinguish the states of interest. In this paper, however, we concentrate on understanding the computational complexity of the \textit{time-bounded quantum state testing} problem. 
In particular, we consider the case where the two states of interest are prepared by polynomial-size quantum circuits and subsequently tracing out non-output qubits. 

The \textsc{Quantum State Distinguishability Problem} (\QSDP{}), defined by Watrous~\cite{Wat02}, is a time-bounded state testing problem with respect to the trace distance, serving as the quantum analog of the \textsc{Statistical Difference Problem} introduced by Sahai and Vadhan~\cite{SV97}.
This promise problem is essential in quantum complexity theory and quantum cryptography, closely linked to quantum statistical zero-knowledge (\QSZK{}). 
The input to this problem consists of the description of two polynomial-size quantum circuits $Q_0$ and $Q_1$ that prepare (the purification of) corresponding quantum (mixed) states $\rho_0$ and $\rho_1$, respectively. 
In \QSDP{}, \textit{yes} instances are those in which the trace distance between these states is at least $\alpha$, while \textit{no} instances are those in which the distance is at most $\beta$, where $0 \leq \beta < \alpha \leq 1$. 
Any input circuits that do not fit into either of these categories are considered outside the promise. 
In this paper, we generalize the parameters $\alpha$ and $\beta$ from constant to efficiently computable functions, denoting this version as $\QSDP[\alpha,\beta]$. 

Error reduction for \SDP{}, known as the \textit{polarization lemma}~\cite{SV97}, polarizes the statistical distance between two probability distributions. Put it differently, for any constants $\alpha$ and $\beta$ such that $\alpha^2 > \beta$, the lemma constructs new distributions such that either very far apart (approaching $1$) for \textit{yes} instances or very close (approaching $0$) for \textit{no} instances, thereby reducing errors on both sides. 
By employing the polarization lemma, the \SZK{} containment of $\SDP[\alpha,\beta]$ when $\alpha^2>\beta$, denoted as the \textit{constant polarizing regime}, is established in~\cite{SV97}. Furthermore, an analog of the direct product lemma for the Hellinger affinity leads to error reduction for \StoqMA{} when the error for \textit{yes} instances is negligible~\cite{Liu21}. 

Sahai and Vadhan left an open problem about reducing error parameters $\alpha$ and $\beta$ beyond the constant polarizing regime, specifically considering the \textit{non-polarizing regime} where $\alpha > \beta > \alpha^2$. 
This challenge also extends to the quantum counterpart \QSDP{}. 
Recently, Berman, Degwekar, Rothblum, and Vasudevan~\cite{BDRV19} made significant progress in addressing this problem by examining the limitations of existing polarization approaches. As a result, they extended the \SZK{} containment of \SDP{} beyond the constant polarizing regime:\footnote{As indicated in~\cite{BDRV19}, \SDP{} is in \SZK{} for $\alpha^2(n)-\beta(n) \geq 1/O(\log n)$ by inspecting constructions in~\cite{SV97}. }
\begin{theorem}[Informal of~\cite{BDRV19}]
    \label{thm:informal-SDP-in-SZK}
    $\SDP[\alpha,\beta]$ is contained in \SZK{} under the following parameter regimes for $\alpha$ and $\beta$\emph{:} 
    \begin{enumerate}[label={\upshape(\roman*)}, topsep=0.33em, itemsep=0.33em, parsep=0.33em]
        \item \label{thmitem:SDPinSZK-polarizing} $\alpha^2(n)-\beta(n) \geq 1/\poly(n)$; 
        \item $\alpha(n)-\beta(n)\geq 1/\poly(n)$, and the $\SDP[\alpha,\beta]$ instance satisfies an additional condition involving the statistical distance $(\SD)$ and the triangular discrimination $(\TD)$.\footnote{More precisely, in the \emph{yes} case, the $\SDP[\alpha,\beta]$ instance is denoted by a pair of polynomial-size Boolean circuits $(C_0,C_1)$ that induce the distributions $(p_0,p_1)$. In the \emph{no} case, the instance is denoted by a  different pair $(C'_0,C'_1)$ inducing the distributions $(p'_0,p'_1)$. These distributions satisfy the following conditions: $\SD(p_0,p_1) > \SD(p'_0,p'_1) > \SD^2(p_0,p_1)$ and $\TD(p_0,p_1) > \TD(p'_0,p'_1)$.}
    \end{enumerate}
\end{theorem}

The proof of \Cref{thm:informal-SDP-in-SZK} entails a series of clever reductions to two time-bounded distribution testing problems: the Jensen-Shannon divergence problem (\JSP{}) and the triangular discrimination problem (\TDP{}).\footnote{See \Cref{table:informal-defs-distances} for informal definitions of the Jensen-Shannon divergence and the triangular discrimination.} These distances are a focal point because they capture the limitation of two known approaches to polarization. In particular, the original polarization lemma~\cite{SV97} focuses on reducing errors alternately for \textit{yes} instances (direct product lemma) and \textit{no} instances (XOR lemma). The triangular discrimination not only implies the latter by definition but also improves the former compared to the statistical distance scenario.\footnote{More precisely, given distributions $p_0$ and $p_1$, one can efficiently construct distributions $p'_0$ and $p'_1$ such that $\TD(p'_0,p'_1)=\TD(p_0,p_1)^k$, which is known as the XOR lemma. One can also efficiently construct distributions $p''_0$ and $p''_1$ such that $1-\exp\rbra*{-\delta k/2} \leq \TD\rbra*{p_0^{\otimes k}, p_1^{\otimes k}} \leq 2k\delta$, where $\TD(p_0,p_1) = \delta$. This is referred to as the direct product lemma. Notably, the dependence on $\delta$ in the lower bound improves from $\delta^2$ to $\delta$ compared to the case of the statistical difference. For further details, see~\cite[Section 4.2.2]{BDRV19}.} Additionally, the entropy extraction approach~\cite{GV99} is essentially based on the Jensen-Shannon divergence,\footnote{This connection arises from the fact that the Jensen-Shannon divergence can be interpreted as the (conditional) entropy difference, as indicated implicitly in Salil Vadhan's PhD thesis~\cite{Vad99}.} which can be viewed as a distance version of entropy difference. 

This work addresses a similar challenge in the quantum world. While classical distances often have several quantum counterparts, the trace distance remains the \textit{sole} quantum analog of the statistical distance. Consequently, the polarization lemma almost directly applies to the trace distance within the constant polarizing regime, as noted in~\cite{Wat02}.
In contrast, quantum counterparts of the Jensen-Shannon divergence and the triangular discrimination -- key tools for examining the limitations of existing techniques in polarizing quantum distances -- either have several choices or have not been defined yet. Defining \textit{proper} quantum analogs for \JSP{} and \TDP{} is therefore a nontrivial task, as these analogs may exhibit behavior distinct from their classical counterparts. 

\paragraph{Why do parameter regimes matter?} Beyond our particular motivation to understand the polarization lemma for quantum distances, we would like to emphasize the general importance of the parameter regime. In computational complexity theory, we typically use \textit{worst-case hardness}, where \textit{a few} \Class{}-hard instances are sufficient to identify that a computational problem \prob{} is hard for the class \Class{}. However, to demonstrate a \Class{} containment of \prob{}, we need to show this containment holds for \textit{all} instances of \prob{} with completeness $c$ and soundness $s$ (the acceptance probability for \textit{yes} instances and \textit{no} instances, respectively), such that $c-s$ is at least the designated (usually polynomial) precision. 
 
Otherwise, we may risk having \textit{a somewhat ``fake'' complete problem}. For instance, if a promise problem \prob{} is proven to be \QSZK{}-hard and is contained in \QSZK{} for a certain parameter regime, then we cannot rule out the possibility that parameter regimes not yet known to be in \QSZK{} may be inherently \QIPii{}-hard, where \QIPii{} denotes the class of promise problems that admit two-message quantum interactive proof systems and contains \QSZK{}.
This possibility suggests that \prob{} is not \QSZK{}-complete unless $\QSZK{}=\QIPii{}$. Unfortunately, such parameter regime issues have appeared in many previous works, such as \cite{Wat02,GHMW15}. To the best of our knowledge, the only \QSZK{}-complete problem with a natural parameter regime, such as $\alpha(n)-\beta(n) \geq 1/\poly(n)$, for which the containment holds, is the \textsc{Quantum Entropy Difference Problem} introduced in~\cite{BASTS10}; see also \Cref{subsubsec:QEDP}.

Resolving such parameter regime issues is often \textit{technically challenging}. A specific example is the low-rank variant of \QSDP{}, where the rank of states $\rho_0$ and $\rho_1$ is at most polynomial in $n$. By leveraging rank-dependent inequalities between the trace distance and the Hilbert-Schmidt distance, such as~\cite[Equation 1.31]{AS17} or~\cite[Equation 6]{CCC19}, we can show that this low-rank variant of \QSDP{} is in \BQP{} for certain parameter regime (with polynomial precision). This is achieved through a clever use of the SWAP test~\cite{BCWdW01}, similar to \Cref{subsec:negl-coQSDP-in-PP}. However, a \BQP{} containment of this problem under the natural regime, as established in~\cite{WZ23}, requires more sophisticated techniques. 
 
\subsection{Main results}

\paragraph{Quantum state testing beyond the constant polarizing regime.}
We introduce two time-bounded state testing problems: the \textsc{Quantum Jensen-Shannon Divergence Problem} (\QJSP{}) and the \textsc{Measured Quantum Triangular Discrimination Problem} (\measQTDP{}). \QJSP{} corresponds to the quantum Jensen-Shannon divergence ($\QJS_2$) defined in~\cite{MLP05},\footnote{The notation $\QJS_2$ denotes the quantum Jensen-Shannon divergence defined using the base-$2$ logarithm.} while \measQTDP{} involves a quantum analog of the triangular discrimination ($\measQTD$) to be explained later. The \QSZK{} containments of these problems, as stated in \Cref{thm:QSZK-containments-informal}, also result in improved \QSZK{} containments of \QSDP{}:\footnote{The reader may feel confused with~\cite[Theorem 5.4]{VW16} on the \QSZK{} containment of \QSDP{} which builds upon adapting techniques in~\cite{SV97}. However, it was claimed in~\cite{GV11} that the proof in~\cite{SV97} does extend to the parameter regime of $\alpha^2(n)-\beta(n) \geq 1/\poly(n)$, but this claim was later retracted, see \cite{Goldreich19}.} 

\begin{theorem}[Improved \QSZK{} containments of \QSDP{}, informal]
\label{thm:QSZK-containments-informal}
For time-bounded state testing problems with respect to the quantum Jensen-Shannon divergence and the measured triangular discrimination problem, specifically \QJSP{} and \measQTDP{}, the following holds, where $n$ denotes the number of qubits used by the states $\rho_0$ and $\rho_1$\emph{:} 
\begin{enumerate}[label={\upshape(\roman*)}, topsep=0.33em, itemsep=0.33em, parsep=0.33em]
    \item \label{thmitem:QSDPinQSZK-polarizing} $\QJSP[\alpha,\beta]$ is in \QSZK{} if $\alpha(n)-\beta(n) \geq 1/\poly(n)$. 
    
    Consequently, $\QSDP[\alpha,\beta]$ is in \QSZK{} if $\alpha^2(n)-\sqrt{2\ln{2}} \beta(n) \geq 1/\poly(n)$. 
    \item \label{thmitem:QSDPinQSZK-nonpolarizing} $\measQTDP[\alpha,\beta]$ is in \QSZK{} if $\alpha(n)-\beta(n) \geq 1/O(\log{n})$. 
    
    This containment further implies that $\QSDP[\alpha,\beta]$ is in \QSZK{} if $\alpha(n)-\beta(n)\geq 1/O(\log{n})$ and the $\QSDP[\alpha,\beta]$ instance satisfies an additional condition: let $(\rho_0,\rho_1)$ and $(\rho'_0,\rho'_1)$ denote the two pairs of quantum states prepared by the respective polynomial-size quantum circuits -- that is, the $\QSDP[\alpha,\beta]$ instance -- in the \emph{yes} and \emph{no} cases. Then, these quantum states must satisfy the following condition\emph{:}\footnote{If $\frac{\rho_0+\rho_1}{2}$ is diagonal of full rank (see \Cref{footnote:measQTD-explicit-formula}), it is fairly effortless to find examples by numerical simulations. For instance, $(\rho_0,\rho_1)$ where $\rho_0=\frac{1}{2}(I+\frac{\sigma_X}{7}+\frac{\sigma_Y}{3}+\frac{\sigma_Z}{4})$ and $\rho_1=\frac{1}{2}(I-\frac{\sigma_X}{7}-\frac{\sigma_Y}{3}-\frac{\sigma_Z}{4})$, together with $(\rho'_0,\rho'_1)$ where $\rho'_0=\frac{1}{2}(I-\frac{\sigma_X}{7}-\frac{\sigma_Y}{5}-\frac{\sigma_Z}{6})$ and $\rho'_1=\frac{1}{2}(I+\frac{\sigma_X}{7}+\frac{\sigma_Y}{5}-\frac{\sigma_Z}{6})$.} 
    \[\td(\rho_0,\rho_1) > \td(\rho'_0,\rho'_1) > \td^2(\rho_0,\rho_1) \text{ and } \measQTD(\rho_0,\rho_1) > \measQTD(\rho'_0, \rho'_1).\] 
\end{enumerate}
\end{theorem}

It is noteworthy that both \QJSP{} and \measQTDP{} are \QSZK{}-complete, where $\QJSP[\alpha,\beta]$ and $\measQTDP[\alpha,\beta]$ are \QSZK{}-hard if $\alpha(n) \leq 1-2^{-n^{1/2-\epsilon}}$ and $\beta(n) \geq 2^{-n^{1/2-\epsilon}}$ for sufficiently large $n$ and some constant $\epsilon\in(0,1/2)$. See \Cref{subsec:QSZK-hardness} for formal statements.

Importantly, our definitions of \measQTDP{} and \QJSP{} serve as \textit{proper} quantum analogs of \TDP{} and \JSP{}, respectively. 
The measured quantum triangular discrimination ($\measQTD{}$) exposes the limitation of the original polarization lemma approach~\cite{SV97,Wat02}, specifically achieving a quadratic improvement in the direct product lemma (\Cref{lemma:measQTD-yes-error-reduction}) with a natural inverse-logarithmic promise gap in its \QSZK{} containment. Notably, the \textsc{Quantum Triangular Discrimination Problem} (\QTDP{}), which is defined using a different quantum analog of the triangular discrimination (\QTD{}) to be explained later, \textit{does not} achieve a similar result.\footnote{The promise problem \QTDP{} suffers from the same parameter regime issue as in the \QSDP{} case: $\QTDP[\alpha,\beta]$ is in \QSZK{} only if $\alpha^2(n)-\beta(n) \geq 1/O(\log{n})$. See \Cref{thm:measQTDP-is-QSZK-complete}\ref{thmitem:QTDP-in-QSZK} for the formal statement.}

Our reductions for proving that \QJSP{} is \QSZK{}-complete also yield a simple \QSZK-hardness proof for the \textsc{Quantum Entropy Difference Problem} (\QEDP{}) introduced in~\cite{BASTS10}, as stated in \Cref{corr:simple-QSZKhardness-QEDP}. Consequently, the quantum Jensen-Shannon divergence captures the limitation of the quantum entropy extraction approach to polarization~\cite{BASTS10}. 
Notably, the quantum Jensen–Shannon divergence arises in the well-known Holevo bound~\cite{Holevo73}, and has long been used in the study of quantum communication complexity~\cite{CvDNT13,NS02}.\footnote{See \Cref{remark:QJS-applications} for further details.} 
However, our implication, namely \Cref{thm:QSZK-containments-informal}\ref{thmitem:QSDPinQSZK-polarizing}, is slightly weaker than the classical counterpart in \Cref{thm:informal-SDP-in-SZK}\ref{thmitem:SDPinSZK-polarizing}. This is because quantum analogs of the triangular discrimination exhibit \textit{distinct behavior} from the classical equivalent. 

\paragraph{Easy regimes for the class \QSZK{}.} 
For $\SDP[1-\epsilon,\epsilon]$, when the error parameter $\epsilon$ is at most some inverse-exponential, then this problem is in \PP{}. 
The existence of an oracle separating \SZK{} from \PP{}, as provided in~\cite{BCHTV19}, highlights the difficulty of establishing \SZK{}-hardness for \SDP{} instances that are contained in \PP{} (referred to as the \textit{easy regime}).\footnote{This challenge is due to the need for non-black-box techniques.}
Let \coQSDP{} and \coSDP{} denote the complement of \QSDP{} and \SDP{}, respectively. 
We establish a similar result for $\QSDP[1-\epsilon,\epsilon]$, where these instances become even easier to solve when error-free: 
\begin{theorem}[Easy regimes for \QSZK{}, informal]
\label{thm:easy-instances-QSZK-informal}
    Let $\epsilon(n)$ be an error parameter satisfying $\epsilon(n) \leq 2^{-n/2-1}$. Then, it holds that $\coQSDP[1-\epsilon,\epsilon]$ is in \PP{}.
    Furthermore, $\coQSDP[1,0]$ is in \NQP{} when there is no error.     
\end{theorem}
 
We notice that \NQP{} (defined in~\cite{ADH97,YY99}) serves as a precise variant of \BQP{} with perfect soundness, specifically having an \textit{exact zero} acceptance probability for \textit{no} instances. Furthermore, researchers initially regarded \NQP{} as a quantum analog of \NP{}.\footnote{\NQP{} is incomparable to \QMA{} due to its equivalence to \PreciseQMA{} with perfect soundness~\cite{KMY09}. Two main distinctions between these classes are: (1) \NQP{} allows an exponentially small gap between acceptance probabilities for \textit{yes} and \textit{no} instances, while \QMA{} permits only an inverse-polynomial gap; and (2) \NQP{} guarantees rejection for \textit{no} instances, whereas \QMA{} allows any reasonable choice. \label{footnote:NQP-vs-QMA}} Prior works~\cite{FGHP99,YY99} have established the relationships $\NQP = \coCeP \subseteq \PP$.

\begin{table}[!htp]
\centering
\begin{tabular}{ccc}
    \toprule
    Parameter regimes & $\coSDP[1-\epsilon,\epsilon]$ & $\coQSDP[1-\epsilon,\epsilon]$ \\
    \midrule
    \multirow{2}{*}{$\epsilon=0$} & in \NP{} & in \NQP{}\\
    & \footnotesize{Folklore} & \footnotesize{This work (\Cref{thm:easy-regimes-for-QSZK}\ref{thmitem:QSZK-easy-regimes-NQP})}\\
    \midrule
    \multirow{2}{*}{$\epsilon(n) \leq 2^{-n/2-1}$} & in \PP{} & in \PP{}\\
    & \footnotesize{Theorem 7.1 in~\cite{BCHTV19}} & \footnotesize{This work (\Cref{thm:easy-regimes-for-QSZK}\ref{thmitem:QSZK-easy-regimes-PP})}\\
    \midrule
    \multirow{2}{*}{$\epsilon(n) \geq 2^{-n^{1/2-\gamma}}$ for $\gamma \in (0,1/2)$} & \SZK{}-hard & \QSZK{}-hard\\
    & \footnotesize{Implicitly stated in~\cite{SV97}} & \footnotesize{Implicitly stated in~\cite{Wat02}}\\
    \bottomrule
\end{tabular}
\caption{Easy and hard regimes for \SZK{} and \QSZK{}.}
\label{table:easy-regimes-QSZK}
\end{table}

We list our results and compare them with the counterpart \SZK{} results in \Cref{table:easy-regimes-QSZK}. 
The improved \SZK{}-hardness and \QSZK{}-hardness follow from skillfully applying the polarization lemma for the relevant distance, as in~\cite[Theorem 3.14]{BDRV19}. 
To demonstrate the \PP{} containment, we first observe $\frac{1}{2}\HS^2(\rho_0,\rho_1)=\frac{1}{2}(\Tr(\rho_0^2)+\Tr(\rho_1^2))-\Tr(\rho_0\rho_1)$. 
The remaining results are mainly derived from a hybrid algorithm based on the SWAP test~\cite{BCWdW01}, namely tossing two random coins and performing the SWAP test on the corresponding states. 

In essence, the phenomenon that parameter regimes with some negligible errors are easier to solve is not unique to \QSZK{}. Analogous phenomena can also be observed in other  complexity classes, such as $\mathsf{QMA(2)}$~\cite{KMY09} and $\mathsf{StoqMA}$~\cite{AGL20}. Nevertheless, it is worth noting that these similar results in other classes do not always necessitate the dimension-preserving property. In particular, polarization lemma for some quantum distance is considered \textit{dimension-preserving} if the resulting quantum states use the same number of qubits as the original quantum states.\footnote{Current techniques for polarizing quantum distances, such as~\cite[Section 4.1]{Wat02} or \Cref{lemma:measQTD-polarization,lemma:QTD-polarization} in this work, increase the number of qubits in the resulting states from $n$ to $\poly(n)$, where $n$ is the number of qubits in the original states, thereby increasing the dimension from $2^n$ to $2^{\poly(n)}$.} 
Since \SZK{} is a subclass of \QSZK{}, \Cref{thm:easy-instances-QSZK-informal} suggests that \QSDP{} may not remain \QSZK{}-hard when the acceptance probability deviates \textit{tinily} from $0$ or $1$.

\subsection{Proof techniques}
\label{sec:intro-proof-techniques}

The \QSZK{} completeness of \QJSP{} and \measQTDP{} crucially relies on inequalities between quantum analogs of common classical $f$-divergences.\footnote{An $f$-divergence is a function $\D_f(p_0\|p_1)$ that measures the difference between two probability distributions $p_0$ and $p_1$, and is defined as $\D_f(p_0\|p_1) \coloneqq \bbE_{x\sim p_1} f(p_0(x)/p_1(x))$. } We start by reviewing and defining these quantum analogs. The most widely used quantum distances are the trace distance ($\td$) and the Bures distance ($\B$, essentially the fidelity), which are quantum counterparts of the statistical distance ($\SD$) and the Hellinger distance ($\H$), respectively. Other commonly used $f$-divergences are the KL divergence (also known as the relative entropy) and the $\chi^2$-divergence, which are unbounded, so we instead focus on their symmetrized versions: the Jensen-Shannon divergence ($\JS$) and the triangular discrimination ($\TD$), respectively. 

The relationship between two quantum analogs of the Jensen-Shannon divergence is a specific instance of Holevo's bound: the measured quantum Jensen-Shannon divergence ($\measQJS$) does not exceed the quantum Jensen-Shannon divergence (\QJS{}). For clarity, we provide informal definitions of these classical and quantum distances in \Cref{table:informal-defs-distances}.\footnote{For formal definitions and additional properties of these classical and quantum distances, we refer the reader to \Cref{subsec:prelim-classical-distances,subsec:prelim-quantum-distances}, respectively.}
\begin{table}[H]
\centering
\adjustbox{max width=\textwidth}{
\begin{tabular}{ccc}
    \toprule
    & Classical (distributions $p_0$ and $p_1$) & Quantum (states $\rho_0$ and $\rho_1$)\\
    \midrule
    Statistical distance & $\SD(p_0,p_1) = \frac{1}{2} \sum_x \abs*{p_0(x)-p_1(x)}$ & $\td(\rho_0,\rho_1) = \frac{1}{2}\Tr\abs*{\rho_0-\rho_1}$ \\
    \midrule
    Hellinger distance & $\Hsquare(p_0,p_1) = 1 - \sum_x \sqrt{p_0(x)}\sqrt{p_1(x)}$ & $\Bsquare(\rho_0,\rho_1) = 2\rbra*{1- \Tr\abs*{\sqrt{\rho_0}\sqrt{\rho_1}}}$ \\
    \midrule
    \rule{0em}{1.25em}Jensen-Shannon & \multirow{2}{*}{$\JS(p_0,p_1) = \H\rbra*{\frac{p_0+p_1}{2}} - \frac{\H(p_0)+\H(p_1)}{2}$} & $\QJS(\rho_0,\rho_1) = \S\rbra*{\frac{\rho_0+\rho_1}{2}} - \frac{\S(\rho_0)+\S(\rho_1)}{2}$ \\
    \rule{0em}{1.25em}divergence & & $\measQJS(\rho_0,\rho_1) = \sup\limits_{{\rm POVM}~\calE} \cbra*{ \JS\rbra*{p_0^{(\calE)}, p_1^{(\calE)}}}$ \\
    \midrule
    \rule{0em}{0.8em}Triangular & \multirow{2}{*}{$\TD(p_0,p_1) = \frac{1}{2}\sum_x \frac{\rbra*{p_0(x)-p_1(x)}^2}{p_0(x)+p_1(x)}$} & \multirow{2}{*}{Not previously known}\\
   \rule{0em}{0.8em} discrimination & & \\
    \bottomrule
\end{tabular}
}
\begin{tablenotes}
    \footnotesize
    \item $\H(p)$ and $\S(\rho)$ denote the Shannon entropy and the von Neumann entropy, respectively.
\end{tablenotes}
\caption{Informal definitions of known classical and quantum distances.}
\label{table:informal-defs-distances}
\end{table}

To the best of our knowledge, there is no known quantum analog of the triangular discrimination. Motivated by its connection to the $\chi^2$-divergence and the family of quantum $\chi^2$-divergence introduced by~\cite{TKRWV10}, we propose the following definitions of the quantum triangular discrimination (\QTD{}) and the measured quantum triangular discrimination ($\measQTD$):
\begin{align*}
    \QTD(\rho_0,\rho_1) &= \frac{1}{2} \Tr\rbra*{(\rho_0-\rho_1) (\rho_0+\rho_1)^{-1/2}(\rho_0-\rho_1)(\rho_0+\rho_1)^{-1/2}},\\
    \measQTD(\rho_0,\rho_1) &= \sup_{{\rm POVM}~\calE} \cbra*{ \TD\rbra*{p_0^{(\calE)}, p_1^{(\calE)}}}.
\end{align*}

We then examine how these quantities relate to the other quantum distances and divergences introduced above:
\begin{theorem}[Inequalities on quantum analogs of the triangular discrimination, informal]
    \label{thm:QTD-inequalities-informal}
    For any quantum states $\rho_0$ and $\rho_1$, we have the following: 
    \begin{enumerate}[label={\upshape(\roman*)}, topsep=0.33em, itemsep=0.33em, parsep=0.33em]
        \item $\td^2 (\rho_0, \rho_1) \leq \measQTD(\rho_0, \rho_1) \leq \QTD(\rho_0, \rho_1) \leq \td(\rho_0, \rho_1)$; 
        \item \label{thmitem:QTDinequalities-QTD-vs-QJS} $\frac{1}{2}\QTD^2(\rho_0,\rho_1) \leq \QJS(\rho_0,\rho_1) \leq \QTD(\rho_0,\rho_1)$; 
        \item $\frac{1}{2}\Bsquare(\rho_0,\rho_1) \leq \measQTD(\rho_0,\rho_1) \leq \Bsquare(\rho_0,\rho_1)$ and $\frac{1}{2}\Bsquare(\rho_0,\rho_1) \leq \QTD(\rho_0,\rho_1) \leq \B(\rho_0,\rho_1)$. 
    \end{enumerate}
\end{theorem}

We summarize our new results and known inequalities in \Cref{table:comparing-distances}, as well as how we utilize these inequalities in our proof. In addition, we highlight that the quantum triangular discrimination behaves differently from its classical counterpart since the triangular discrimination is a constant multiplicative error approximation of the Jensen-Shannon divergence. This difference breaks down the quantum equivalent of the ingenious reduction from \TDP{} to \JSP{} presented in~\cite{BDRV19}, leading to a slightly worse parameter in the improved \QSZK{} containment of \QSDP{}, as stated in \Cref{thm:QSZK-containments-informal}\ref{thmitem:QSDPinQSZK-polarizing}. 

\begin{table}[H]
\centering
\adjustbox{max width=\textwidth}{
\begin{tabular}{cccc}
    \toprule
     & Classical & Quantum & Usages related to \QSZK{}\\
    \midrule
    \multirow{2}{*}{$\SD \text{ vs. } \Hsquare$} & $\Hsquare \leq \SD \leq \sqrt{2}\H$ & $\frac{1}{2}\B^2 \leq \td \leq \B$ & A polarization lemma for\\
    & \footnotesize{\cite{Kailath67}} & \footnotesize{\cite{FvdG99}} & the trace distance \footnotesize{\cite{Wat02}}\\
    \midrule
    \multirow{2}{*}{$\SD \text{ vs. } \JS$} & $1\!-\!\binH\big(\frac{1-\SD}{2}\big) \leq \JS_2 \leq \SD$ & $1\!-\!\binH\big(\frac{1-\td}{2}\big) \leq \QJS_2 \leq \td$ & \QJSP{} is \QSZK{}-hard\\
    & \footnotesize{\cite{FvdG99,Top00}} & \footnotesize{\cite{BH09,FvdG99}} & \footnotesize{This work (\Cref{lemma:QJSP-is-QSZK-hard})}\\
    \midrule
    \multirow{3}{*}{$\SD \text{ vs. } \TD$} & \multirow{2}{*}{$\SD^2 \leq \TD \leq \SD$} & \multirow{2}{*}{$\td^2 \leq \measQTD \leq \QTD \leq \td$} & \measQTDP{} and \QTDP{}\\
    & & & are \QSZK{}-hard\\
    & \footnotesize{\cite{Top00}} & \footnotesize{This work (\Cref{thm:QTD-vs-td})} & \footnotesize{This work (\Cref{lemma:measQTDP-is-QSZK-hard})}\\
    \midrule
    \multirow{2}{*}{$\JS \text{ vs. } \TD$} &  $\frac{1}{2} \TD \leq \JS \leq \ln{2}\cdot \TD$ & $\frac{1}{2}\QTD^2 \leq \QJS \leq \QTD$ & \multirow{2}{*}{None}\\
    & \footnotesize{\cite{Top00}} & \footnotesize{This work (\Cref{thm:QTD-vs-QJS})} & \\
    \midrule
    \multirow{3}{*}{$\TD \text{ vs. } \Hsquare$} & \multirow{2}{*}{$\Hsquare \leq \TD \leq 2\Hsquare$} & $\frac{1}{2}\Bsquare \leq \measQTD \leq \Bsquare$ & Polarization lemmas \\
    & & $\frac{1}{2} \Bsquare \leq \QTD \leq \B$ & for $\measQTD$ and \QTD{}\\
    & \footnotesize{\cite{LeCam86}} & \footnotesize{This work (\Cref{thm:QTD-vs-B})} & \footnotesize{This work (\Cref{lemma:measQTD-polarization,lemma:QTD-polarization})}\\
    \bottomrule
\end{tabular}
}
\caption{A comparison between classical and quantum distances with usages related to \QSZK{}.}
\label{table:comparing-distances}
\end{table}

Leveraging inequalities in \Cref{table:comparing-distances}, we establish that \QJSP{}, \measQTDP{}, and \QTDP{} are \QSZK{}-complete. 
The \QSZK{} containments of \measQTDP{} and \QTDP{} are achieved through \textit{new} polarization lemmas for the measured quantum triangular discrimination ($\measQTD$) and the quantum triangular discrimination (\QTD{}), while we establish the \QSZK{} containment of \QJSP{} via a reduction to \QEDP{}~\cite{BASTS10} using the joint entropy theorem on classical-quantum states. We thus explore the limitations of current techniques to quantum polarize quantum distances. Additionally, the \QSZK{}-hardness of these problems is directly analogous to their classical counterparts~\cite{BDRV19} because of the corresponding inequalities in \Cref{table:comparing-distances}. 

\subsection{Discussion and open problems} 

\paragraph{Better upper bounds for \textsc{GapQSD} and \QSZK{}.} The best known upper bound for the promise problem \textsc{GapQSD}, specifically $\QSDP[\alpha,\beta]$ when $\alpha(n)-\beta(n) \geq 1/\poly(n)$, is \QIPii{}, as shown implicitly in~\cite{Wat02,JUW09}. More recently, this upper bound was slightly improved to \QIPii{} with a quantum single-exponential-time and linear-space honest prover in~\cite{LGLW23}, which appeared after the release of our work. Since the classical counterpart \textsc{GapSD} is contained in $\AMcapcoAM$~\cite{BL13}, it is natural to ask whether the upper bound for \GapQSD{} (or \QSZK{}) can be improved further, perhaps to subclasses of \QIPii{} in which the verifier's message has some particular form, as introduced in~\cite{MW05,KLGN19}?

\paragraph{Applications of quantum analogs of the triangular discrimination.} Is there any other application of the (measured) quantum triangular discrimination besides its usage on \QSZK{}? For instance, Yehudayoff~\cite{Yehudayoff20} utilized triangular discrimination to obtain a sharper communication complexity lower bound of the point chasing problem. Can we expect a similar implication in the quantum world?
Moreover, we note that $\measQTD$ is a symmetric version of the measured Bures $\chi^2$-divergence and the latter is used for the nonzero testing of quantum mutual information~\cite{FO23}. Might $\measQTD$ also play a role in quantum property testing? 

\paragraph{Improved inequalities on the quantum triangular discrimination. }
We observe that \Cref{thm:QTD-inequalities-informal}\ref{thmitem:QTDinequalities-QTD-vs-QJS} is not \textit{tight}. Numerical simulations indicate that the tight bound is $\QTD^2(\rho_0,\rho_1) \leq \QJS_2(\rho_0,\rho_1) \leq \QTD(\rho_0,\rho_1)$ for any states $\rho_0$ and $\rho_1$. 
This bound can be saturated by choosing states $\rho_0$ and $\rho_1$ with orthogonal support, which suffices to make $\QJS_2(\rho_0,\rho_1)$ and $\QTD(\rho_0,\rho_1)$ equal to $1$. 
Furthermore, numerical simulations also suggest that the triangle inequality holds for the square root of \QTD{}, namely $\sqrt{\QTD(\rho_0,\rho_1)} + \sqrt{\QTD(\rho_1,\rho_2)} \geq \sqrt{\QTD(\rho_0,\rho_2)}$ for any states $\rho_0,\rho_1$, and $\rho_2$. 
This indicates that $\sqrt{\QTD}$ is a \textit{metric}, with the same property also holding for triangular discrimination~\cite{LeCam86}. 

\subsection{Related works and recent developments}
Beyond the trace distance and quantum analogs of triangular discrimination, approaches similar to the original polarization lemma have been extended to other quantum settings. One example is quantum channel testing (equivalently, distinguishing mixed-state quantum circuits) with respect to the diamond norm distance, as introduced in~\cite{RW05}, where the associated promise problem, \textsc{Quantum Circuit Distinguishability} (\textsc{QCD}), has been shown to be \QIP{}-complete. The diamond norm distance between two quantum channels $\Phi_0$ and $\Phi_1$ is defined as $\norm{\Phi_0-\Phi_1}_{\diamond} \coloneqq \sup_{\rho} 2 \cdot \td\rbra*{\rbra*{\Phi_0\otimes I}(\rho), \rbra*{\Phi_1\otimes I}(\rho)}$. 
More recently (after our work was released), quantum state testing with respect to the quantum $\ell_{\alpha}$ distance for $1 < \alpha(n) \leq 1+1/n$, as considered in~\cite{LW25Lalpha}, a generalization of the trace distance ($\alpha=1$) through the Schatten $\alpha$-norm $\norm*{A}_{\alpha} \coloneqq \Tr\rbra*{|A|^{\alpha}}^{1/\alpha}$, has been proven to be \QSZK{}-complete. Notably, the \QSZK{} containment also holds throughout the polarizing regime. 

Additionally, the entropy extraction approach to polarizing quantum distances has been used to establish a complete problem for the class \textsf{qq-QAM}, introduced in~\cite{KLGN19}. This class, which is a subclass of \QIPii{}, consists of promise problems that admit two-message quantum interactive proof systems where the verifier's message is restricted to halves of EPR pairs. This complete problem, \textsc{Maximum Output Quantum Entropy Approximation} (\textsc{MaxOutQEA}), is defined in terms of the maximum output von Neumann entropy of a quantum channel $\Phi$, given by $\S_{\max}(\Phi) \coloneqq \max_{\rho} \S(\Phi(\rho))$, where $\S(\cdot)$ denotes the von Neumann entropy. 

Very recently, building on the strategy of our simple \QSZK{}-hardness proof for \QEDP{} (\Cref{corr:simple-QSZKhardness-QEDP}), the \textsc{Quantum $q$-Tsallis Entropy Approximation Problem} (\textsc{TsallisQEA}$_q$) and its entropy difference version were shown to be \BQP{}-complete for constant $q > 1$ in~\cite{LW25entropy}, using a generalized notion of the quantum Jensen-Shannon divergence. The quantum $q$-Tsallis entropy is defined as $\S_q(\rho) \coloneqq \frac{1-\Tr\rbra*{\rho^q}}{q-1}$.\footnote{The quantum $q$-Tsallis entropy $\S_q(\rho)$ converges to the von Neumann entropy $\S(\rho)$ as $q$ approaches $1$.} When $q=2$, this result implies the \BQP{}-hardness of the purity testing problem (\textsc{Purity}), which asks whether $\Tr(\rho^2)$ is at least $2/3$ or at most $1/3$. Although \BQP{} containment of \textsc{Purity} has been known for over two decades via the SWAP test~\cite{BCWdW01}, proving \BQP{}-hardness remained open until this recent work. 

%%%%%%%%%%%%%%%%%%%%%%%%%%%%%%%%%%%%%%%%%%%%%%%%%%%%%

\section{Preliminaries}
\label{sec:preliminary}

\subsection{Distances and divergences for classical probability distributions}
\label{subsec:prelim-classical-distances}

In this subsection, we will review several commonly used classical distances and divergences. 
We begin by defining the statistical distance and the triangular discrimination.  

\begin{definition}[Statistical distance]
	\label{def:statDist}
	The statistical distance between two probability distributions $p_0$ and $p_1$ on $\calX$ is defined by $\SD(p_0,p_1)  \coloneqq  \frac{1}{2}\|p_0-p_1\|_1 = \frac{1}{2}\sum_{x\in\calX} |p_0(x)-p_1(x)|$. 
\end{definition}

\begin{definition}[Triangular discrimination]
	\label{def:triDist}
	The triangular discrimination, also known as the Le Cam divergence, between two probability distributions $p_0$ and $p_1$ on $\calX$ is defined by 
    \[\TD(p_0,p_1)  \coloneqq  \frac{1}{2}\sum_{x \in \calX} \frac{(p_0(x)-p_1(x))^2}{p_0(x)+p_1(x)}.\] 
\end{definition}
It is noteworthy that $\TD$ is a symmetrized version of the $\chi^2$ divergence, namely $\TD(p_0,p_1) = \chi^2\!\left(p_z \big\| \frac{p_0+p_1}{2}\right)$ for $z \in \binset$. We also know that $\SD^2(p_0,p_1) \leq \TD(p_0,p_1) \leq \SD(p_0,p_1)$, as presented in \cite{Top00}.
Next, we proceed by defining the Jensen-Shannon divergence ($\JS$).  

\begin{definition}[Jensen-Shannon divergence]
	\label{def:JSdiver}
	The Jensen-Shannon divergence between two probability distributions $p_0$ and $p_1$ is defined by
	$\JS_2(p_0,p_1) \coloneqq \H\!\left(\frac{p_0+p_1}{2}\right) - \frac{1}{2}\rbra*{\H(p_0)+\H(p_1)}$,
	where the Shannon entropy $\H(p) \coloneqq -\sum_{x}p(x) \log_2 p(x)$.  	
\end{definition}

The Jensen-Shannon divergence serves as a symmetrized version of the Kullback–Leibler divergence (also known as relative entropy), namely, $\JS_2(p,q) = \frac{1}{2} {\rm KL}\left(p \big\Vert \frac{p+q}{2} \right) + \frac{1}{2} {\rm KL}\left(q \big\Vert \frac{p+q}{2} \right)$, which follows from a straightforward calculation. 

\begin{proposition}[Adapted from \cite{FvdG99,Top00}]
	\label{prop:JS-vs-SD}
    For any two probability distributions $p_0$ and $p_1$, the following inequalities hold\emph{:}
	\[\sum_{v=1}^{\infty} \frac{\SD(p_0,p_1)^{2v}}{\ln{2}\cdot 2v(2v-1)}=1-\binH\left(\frac{1-\SD(p_0,p_1)}{2}\right) \leq \JS_2(p_0,p_1) \leq \SD(p_0,p_1),\] 
	where the binary entropy function $\binH(x) \coloneqq -x\log_2(x)-(1-x)\log_2(1-x)$. 
\end{proposition}

Finally, we define the Hellinger distance and the inner product $\langle P|Q \rangle$ between normalized non-negative vectors, also known as \textit{Hellinger affinity} (or \textit{Bhattacharyya coefficient}). 
\begin{definition}[Squared Hellinger distance]
	\label{def:HelDist}
	The squared Hellinger distance between two probability distributions $p_0$ and $p_1$ on $\calX$ is defined by
	$\Hsquare(p_0,p_1)  \coloneqq  \frac{1}{2} \sum_{x \in \calX} (\sqrt{p_0(x)}-\sqrt{p_1(x)})^2 = 1-\innerprod{P_0}{P_1}$, where $\ket{P_0} \coloneqq \sum_{x}\sqrt{p_0(x)}\ket{x}$ and $\ket{P_1} \coloneqq \sum_{x} \sqrt{p_1(x)}\ket{x}$. 
\end{definition}

Additionally, a simple observation~\cite[Page 48]{LeCam86} indicates the squared Hellinger distance is very close to the triangular discrimination, namely $\Hsquare(p_0,p_1) \leq \TD(p_0,p_1) \leq 2\Hsquare(p_0,p_1)$. 

\subsection{Distances and divergences for quantum states}
\label{subsec:prelim-quantum-distances}

Now we will review relevant quantum distances and divergences. 
We say that a square matrix $\rho$ is a quantum state if $\rho$ is a positive semi-definite and has trace one.
Classical distances and divergences often have corresponding quantum versions, and sometimes even \textit{multiple options}. 
These distances usually reduce to the classical counterpart when quantum states $\rho_0=\diag(p_0)$ and $\rho_1=\diag(p_1)$ are diagonal. We recommend~\cite[Section 3.1]{BOW19} for a nice survey. 

For any classical $f$-divergence, a quantum analog can be defined in one of two ways: either by converting arithmetic operations in the classical divergence to their matrix-theoretic counterparts, or by considering the probability distributions obtained by applying the same positive operator-valued measure (POVM) to both quantum states, a quantity referred to as the \textit{measured quantum $f$-divergence}. For a comprehensive overview of other quantum analogs of $f$-divergence and their relationships, we refer to \cite{Hiai21}. Given a classical $f$-divergence $\mathrm{d}_f(\cdot,\cdot)$, the corresponding measured quantum $f$-divergence $\mathrm{D}_f(\cdot,\cdot)$ is defined as follows: 
\begin{equation}
\label{eq:measured-f-divergences}
\mathrm{D}_f(\rho_0,\rho_1) = \sup_{\mathrm{POVM}~\calE} \left\{\mathrm{d}_f\left( p_0^{(\calE)}, p_1^{(\calE)} \right)\right\}, \text{ where } p_z^{(\calE)} \coloneqq \rbra*{\Tr(\rho_z E_1), \cdots, \Tr(\rho_z E_N)}.
\end{equation}
Here, $z\in\binset$ and $N$ denotes the dimension of the quantum states $\rho_0$ and $\rho_1$. Put it differently, this quantum divergence can be viewed as the maximum classical divergence that is achievable when the same POVM is applied to both quantum states. 

\paragraph{Quantum analogs of statistical distance and Hellinger distance. }
\label{subsec:trace-distance-and-Bures}
We start with the trace distance, which is a metric. This distance has a maximum value of $1$ occurring when $\rho_0$ and $\rho_1$ have orthogonal supports. Moreover, the trace distance is a measured version of the statistical distance in terms of \Cref{eq:measured-f-divergences}, as stated in, e.g.,~\cite[Theorem 9.1]{NC10}. 
\begin{definition}[Trace distance]
    The trace distance between two quantum states $\rho_0$ and $\rho_1$ is defined by  
    $\td(\rho_0,\rho_1) \coloneqq \frac{1}{2}\Tr|\rho_0-\rho_1| = \frac{1}{2}\Tr(((\rho_0-\rho_1)^\dagger(\rho_0-\rho_1))^{1/2}).$
\end{definition}

Although the squared Hellinger distance is closely related to the inner product, there are several quantum analogs because of the non-commuting nature of matrices. 
Based on the (Uhlmann) fidelity, we proceed with the \textit{squared Bures distance}, which is the first quantum analog of \Cref{def:HelDist} since it is precisely the \textit{measured squared Hellinger distance}~\cite{FC94}: 
\begin{definition}[Squared Bures distance]
    \label{def:BuresDist}
    The squared Bures distance between two quantum states $\rho_0$ and $\rho_1$ is defined as $\Bsquare(\rho_0,\rho_1) \coloneqq 2(1-\F(\rho_0,\rho_1))$, where $\F(\rho_0,\rho_1) \coloneqq \Tr|\sqrt{\rho_0}\sqrt{\rho_1}|$ denotes the fidelity between $\rho_0$ and $\rho_1$.
\end{definition}

We also provide inequalities between the trace distance and the Bures distance. 
\begin{proposition}[Adapted from~\cite{FvdG99}]
\label{prop:traceDist-vs-B}
For any quantum states $\rho_0$ and $\rho_1$, \[\frac{1}{2}\Bsquare(\rho_0,\rho_1) \leq \td(\rho_0,\rho_1) \leq \B(\rho_0,\rho_1).\]
\end{proposition}

Notice that the matrices $(ABA)^{1/2}$ and $A^{1/2}B^{1/2}A^{1/2}$ are not generally equal. This fact suggests that the Uhlmann fidelity differs from the quantum Hellinger affinity $\mathrm{Q}_{1/2}(\rho_0,\rho_1) \coloneqq \Tr(\sqrt{\rho_0}\sqrt{\rho_1})$. The latter gives rise to the second quantum analog of \Cref{def:HelDist}:

\begin{definition}[Quantum squared Hellinger distance]
    For any states $\rho_0$ and $\rho_1$, quantum squared Hellinger distance is defined by $\QHD(\rho_0,\rho_1) \coloneqq \frac{1}{2}\Tr(\sqrt{\rho_0}-\sqrt{\rho_1})^2=1-\mathrm{Q}_{1/2}(\rho_0,\rho_1).$
\end{definition}

Additionally, it is noteworthy that $\mathrm{F}(\rho_0,\rho_1) \geq \mathrm{Q}_{1/2}(\rho_0,\rho_1)$. We recommend two comprehensive reviews \cite{CS20,BGJ19} for summarizing different variants of the fidelity. 

\paragraph{Quantum analogs of Jensen-Shannon Divergence.}
\label{subsec:quantum-JS}
We will encounter various quantum analogs of the Jensen-Shannon divergence. 
The study of quantum analogs of the Jensen-Shannon Divergence (also known as Shannon Distinguishability) can be traced back to the well-known Holevo bound~\cite{Holevo73}. We begin with the definition given in~\cite{MLP05} and note that $\QJS_2$ is at most $1$,\footnote{Whereas quantum relative entropy is unbounded, by properties of von Neumann entropy, such as~\cite[Theorem 11.8]{NC10}, we know that $\QJS_2(\rho_0,\rho_1) \leq \binH(1/2)=1$. Then the equality holds if and only if $\rho_0$ and $\rho_1$ have support on orthogonal subspaces.} where the subscript of $2$ indicates that it is defined using the base-$2$ logarithm:\footnote{The same subscript $2$ convention also applies to both the Jensen-Shannon divergence ($\JS_2$) and the measured quantum Jensen-Shannon divergence ($\QJS_2$), which will be defined later.}

\begin{definition}[Quantum Jensen-Shannon Divergence, adapted from~{\cite[Section III]{MLP05}}]
	\label{def:QJS}
	The quantum Jensen-Shannon divergence between two quantum states $\rho_0$ and $\rho_1$ is defined by 
    \[\QJS(\rho_0,\rho_1)  \coloneqq  \S\rbra*{\frac{\rho_0+\rho_1}{2}} -\frac{\S(\rho_0)+\S(\rho_1)}{2}=\frac{1}{2}\sbra*{ \D\rbra[\bigg]{ \rho_0\bigg\| \frac{\rho_0+\rho_1}{2} } + \D\rbra[\bigg]{ \rho_1\bigg\| \frac{\rho_0+\rho_1}{2} } }.\]
	Here, $\S(\rho) \coloneqq -{\rm Tr}(\rho \ln \rho)$ denotes the von Neumann entropy of the quantum state $\rho$, and $\D(\rho_0 \| \rho_1)$ denotes the quantum relative entropy between $\rho_0$ and $\rho_1$. 
\end{definition} 

It is worth noting that the square root of the quantum Jensen-Shannon divergence was recently proven to be a metric~\cite{Virosztek21,Sra21}, and thus satisfies the triangle inequality. 
In addition, the measured variant of the Jensen-Shannon divergence between quantum states $\rho_0$ and $\rho_1$, which aligns with \Cref{eq:measured-f-divergences} and is also known as \textit{quantum Shannon distinguishability}, was studied by Fuchs and van de Graaf~\cite{FvdG99}. This quantity, referred to as the \textit{Measured Quantum Jensen-Shannon Divergence} and denoted by $\measQJS(\rho_0,\rho_1)$ in this work, does not have an explicit formula, as it serves as a solution to some transcendental equation~\cite{FC94}.

\begin{remark}[Applications of the quantum Jensen-Shannon divergence]
    \label{remark:QJS-applications}
    The quantum Jensen-Shannon divergence (\QJS{}) coincides with a special case of the right-hand side of the well-known Holevo bound~\cite{Holevo73}, such as \cite[Theorem 12.1]{NC10}, specifically the Holevo $\chi$ quantity for size-$2$ ensembles with a uniform distribution. Furthermore, because the Holevo bound can also be used to bound the amount of (quantum) communication between two parties who may share entanglements~\cite{NS02,CvDNT13},\footnote{See also~\cite[Section 15.2]{deWolf19} for a pedagogical overview.} the quantum Jensen-Shannon divergence has implicitly played a role in the study of quantum communication complexity. 
\end{remark}

Moreover, the quantum Jensen-Shannon divergence is upper-bounded by the trace distance, as shown in \Cref{lemma:QJS-leq-traceDist}. The proof of this lemma essentially adapts the construction used to establish an analogous bound for classical distributions, such as~\cite[Claim 4.4.2]{Vad99}.

\begin{lemma}[Adapted from~{\cite[Theorem 14]{BH09}}]
    \label{lemma:QJS-leq-traceDist}
    For any quantum states $\rho_0$ and $\rho_1$, 
    \[\QJS(\rho_0,\rho_1) \leq \ln{2}\cdot \td(\rho_0,\rho_1).\]
\end{lemma}

\begin{proof}
We begin with the construction in~\cite[Theorem 14]{BH09}. Consider a single qutrit register $B$ with basis vectors $\ket{0},\ket{1},\ket{2}$. 
Define $\tilde{\rho}_0$ and $\tilde{\rho}_1$ on $\calH\otimes \calB$ as below, where $\calB=\bbC^3$ is the Hilbert space corresponding to the register $B$:
\begin{align*}
    \tilde{\rho}_0  \coloneqq & \frac{\rho_0+\rho_1-|\rho_0-\rho_1|}{2} \otimes \ketbra{2}{2} + \frac{\rho_0-\rho_1+|\rho_0-\rho_1|}{2} \otimes \ket{0}\bra{0}  \coloneqq \sigma_2\otimes\ketbra{2}{2} + \sigma_0\otimes\ket{0}\bra{0},\\
    \tilde{\rho}_1  \coloneqq & \frac{\rho_0+\rho_1-|\rho_0-\rho_1|}{2} \otimes \ketbra{2}{2} + \frac{\rho_1-\rho_0+|\rho_0-\rho_1|}{2} \otimes \ket{1}\bra{1}  \coloneqq \sigma_2\otimes\ketbra{2}{2} + \sigma_1\otimes\ket{1}\bra{1}.
\end{align*}

Here, $\sigma_0$ corresponds to the regime that $\rho_0$ is ``larger than'' $\rho_1$ (where $\rho_0$ and $\rho_1$ are ``distinguishable'') and so does $\sigma_1$, whereas $\sigma_2$ corresponds to the regime that $\rho_0$ is ``indistinguishable'' from $\rho_1$. 
One can see this construction generalizes the proof of the classical counterparts, namely~\cite[Claim 4.4.2]{Vad99}, to quantum distances. 

Then it is left to show $\QJS(\rho_0,\rho_1) \leq \QJS(\tilde{\rho}_0,\tilde{\rho}_1) = \td(\rho_0,\rho_1)$. 
Using the data-processing inequality of the quantum relative entropy, such as~\cite[Theorem 3.9]{Petz07}, we obtain
\begin{equation}
    \label{eq:QJS-traceDist-RHS}
    \begin{aligned}
    \QJS(\rho_0,\rho_1) &= \QJS(\Tr_B(\tilde{\rho_0}),\Tr_B(\tilde{\rho_1}))\\
    &\leq \QJS(\tilde{\rho}_0,\tilde{\rho}_1)\\
    &= -\Tr\rbra*{ \frac{\tilde{\rho}_0+\tilde{\rho}_1}{2} \ln\frac{\tilde{\rho}_0+\tilde{\rho}_1}{2} } + \frac{1}{2}\left( \Tr\left(\tilde{\rho}_0 \ln \tilde{\rho}_0\right) + \Tr\left(\tilde{\rho}_1 \ln \tilde{\rho}_1\right) \right).
    \end{aligned}
\end{equation}

Here, the first line is because of $\Tr_B(\tilde{\rho}_k)=\rho_k$ for $k\in\binset$, and the third line owes to $\QJS(\rho_0,\rho_1)=\S(\frac{\rho_0+\rho_1}{2})-\frac{1}{2}(\S(\rho_0)+\S(\rho_1))$ for any quantum states $\rho_0$ and $\rho_1$.  
Noting that $\sigma_0 \otimes \ket{0}\bra{0}$, $\sigma_1 \otimes \ket{1}\bra{1}$, and $\sigma_2 \otimes \ket{2}\bra{2}$ are orthogonal to each other, and $\ln(A+B)=\ln(A)+\ln(B)$ when $A$ and $B$ are orthogonal (i.e., $AB=BA=0$), we have derived that
\begin{equation}
    \label{eq:QJS-traceDist-RHS-terms}
    \begin{aligned}
    \Tr\left(\frac{\tilde{\rho}_0+\tilde{\rho}_1}{2} \ln\frac{\tilde{\rho}_0+\tilde{\rho}_1}{2}\right) 
    &= \Tr(\sigma_2\ln \sigma_2) + \sum_{k\in\binset}\Tr\left(\frac{\sigma_0}{2}\ln\frac{\sigma_0}{2}\right),\\
    \forall k\in\binset,~\Tr\left( \tilde{\rho}_k\ln\tilde{\rho}_k \right) &= \Tr(\sigma_2 \ln \sigma_2) + \Tr(\sigma_k \ln \sigma_k). 
    \end{aligned}
\end{equation}

Combining \Cref{eq:QJS-traceDist-RHS-terms,eq:QJS-traceDist-RHS}, we finish the proof: 
    \[\QJS(\rho_0,\rho_1) 
    \leq \Tr\left[\frac{\sigma_0}{2} \left(\ln \sigma_0 \!-\! \ln\frac{\sigma_0}{2}\right) \right] + \Tr\left[\frac{\sigma_1}{2} \left(\ln \sigma_1 \!-\! \ln\frac{\sigma_1}{2}\right) \right]
    = \frac{\ln{2}}{2}\cdot \Tr(\sigma_0+\sigma_1)
    = \ln{2}\cdot\td(\rho_0,\rho_1),\] 
    where the third equality is due to $\sigma_0+\sigma_1=|\rho_0-\rho_1|$. 
\end{proof}

Owing to the Holevo bound, we know that the quantum Jensen-Shannon divergence is at least its measured variant ($\measQJS$): 

\begin{proposition}[Quantum Jensen-Shannon divergence is at least its measured variant]
    \label{prop:measQJS-leq-QJS}
    For any quantum states $\rho_0$ and $\rho_1$, $\measQJS(\rho_0,\rho_1) \leq \QJS(\rho_0,\rho_1)$.
\end{proposition}

\begin{proof}
We begin by stating an equivalent characterization of the classical Jensen-Shannon divergence, building upon its fundamental property, such as~\cite[Proposition 4.1]{BDRV19}: 
\begin{proposition}[Mutual information interpretation of Jensen-Shannon divergence]
    \label{fact:mutualInfo-JSD}
    For any distributions $p_0$ and $p_1$, let $T$ be a binary indicator variable that chooses the value of $x$ according to $p_i$ if $T=i$ where $i\in\binset$, as well as let $X$ be a random variable associated with a uniform mixture distribution between $p_0$ and $p_1$. Then, we obtain 
    \[\JS(p_0,p_1)=I(T;X)=\H(T)-\H(T|X)=1-\H(T|X).\]
\end{proposition}

Following the observation in \Cref{remark:QJS-applications}, $\QJS(\rho_0,\rho_1)$ corresponds to the Holevo $\chi$ quantity for the ensemble $\{1/2,\rho_0; 1/2,\rho_1\}$. We can also observe that $\measQJS(\rho_0,\rho_1)$ equals the accessible information of the same ensemble. Therefore, this equivalence allows the proof to follow directly from the Holevo bound. 
\end{proof}

Utilizing \Cref{prop:measQJS-leq-QJS}, we obtain a lower bound of \QJS{} in terms of the trace distance:

\begin{lemma}[Adapted from~\cite{Holevo73,FvdG99}]
    \label{lemma:traceDist-leq-QJS}
    For any quantum states $\rho_0$ and $\rho_1$, 
    \[\QJS_2(\rho_0,\rho_1) \geq \measQJS_2(\rho_0,\rho_1) \geq 1-\binH\left(\frac{1-\td(\rho_0,\rho_1)}{2}\right) = \sum\nolimits_{v=1}^{\infty}\frac{\td(\rho_0,\rho_1)^{2v}}{\ln{2}\cdot 2v(2v-1)},\]
    where the binary entropy $\binH(p) \coloneqq -p\log_2(p)-(1-p)\log_2(1-p)$.
\end{lemma}

\begin{proof}
We first fix some POVM measurement 
$\calE=\{E_x\}_{x\in \calU}$ where  $\calU=\supp{\rho_0}\cup\supp{\rho_1}$.
And let $p_z^{(\calE)}$ be the induced distribution with respect to the POVM $\calE$ of $\rho_z$ for $z\in\binset$.  
By utilizing the left-hand side inequality in \Cref{prop:JS-vs-SD}, we have
\begin{equation}
\label{eq:measured-JS-vs-SD}
\measQJS_{\calE^*}(\rho_0,\rho_1) \geq \measQJS_{\calE}(\rho_0,\rho_1) = \JS(p_0^{(\calE)}, p_1^{(\calE)})
\geq \sum_{v=1}^{\infty} \frac{\SD\big(p_0^{(\calE)},p_1^{(\calE)}\big)^{2v}}{2v(2v-1)},
\end{equation}
where $\calE^*$ is an optimal measurement of $\measQJS(\rho_0,\rho_1)$. 
Let $g(x) \coloneqq \sum_{v=1}^{\infty} \frac{x^{2v}}{2v(2v-1)}$, then $g(x)$ is monotonically increasing on $0 \leq x \leq 1$. 
Since \Cref{eq:measured-JS-vs-SD} holds for arbitrary POVM $\calE$, as well as the trace distance serves as the measured version of the statistical distance, we complete the proof by choosing the one that maximizes $\td(\rho_0,\rho_1)$.
\end{proof}

\subsection{Quantum state testing in the trace distance and beyond}

In this subsection, we will provide definitions of the (time-bounded) quantum state testing problem with respect to different distance-like measures, along with several useful results concerning these computational problems. 
We start with a formal definition of the \textsc{Quantum State Distinguishability Problem}, denoted as $\QSDP[\alpha,\beta]$. 

\begin{definition}[Quantum State Distinguishability Problem, $\QSDP{[\alpha,\beta]}$, adapted from~{\cite[Section 3.3]{Wat02}}]
	\label{def:QSDP}
    Let $Q_0$ and $Q_1$ be polynomial-size quantum circuits that act on $m(n)$ qubits and having $n$ specified output qubits, where $m$ is polynomial in $n$. For $i\in\binset$, let $\rho_i$ denote the quantum state obtained by running $Q_i$ on state $\ket{0^n}$ and tracing out the non-output qubits. Let $\alpha$ and $\beta$ denote efficiently computable functions. Then promise that one of the following cases will occur: 
	\begin{itemize}[topsep=0.33em, itemsep=0.33em, parsep=0.33em]
		\item \emph{Yes:} A pair of quantum circuits $(Q_0,Q_1)$ such that $\td(\rho_0,\rho_1) \geq \alpha(n)$; 
		\item \emph{No:} A pair of quantum circuits $(Q_0,Q_1)$ such that $\td(\rho_0,\rho_1) \leq \beta(n)$.
	\end{itemize}
\end{definition}

\begin{remark}[The choice of $n$ in \QSDP{}]
    \label{remark:n-choices}
    The definition of \QSDP{} in \Cref{def:QSDP} matches the counterpart classical promise problem,  particularly \SDP{}, from~\cite[Section 2.2]{SV97}, but it is slightly more restrictive than the version in~\cite[Section 3.3]{Wat02}. In particular, \Cref{def:QSDP} assumes that the input length $m$ and the output length $n$ are \textit{polynomially related}, while the version in~\cite{Wat02} allows the output length to be \textit{much smaller}. Such cases may not remain \QSZK{}-hard, e.g., the variant of \QSDP{} with output length $1$ is \BQP{}-complete, as observed in~\cite[Theorem 9]{Kobayashi03}.     
\end{remark}

In analogy with \Cref{def:QSDP}, we can define the \textsc{Quantum Jensen-Shannon Divergence Problem} (\QJSP{}), the \textsc{Measured Quantum Triangular Discrimination Problem} (\measQTDP{}), and the \textsc{Quantum Triangular Discrimination Problem} (\QTDP{}) by replacing the underlying closeness measure as follows:
\begin{itemize}
    \item $\QJSP[\alpha,\beta]$: Decide whether $\QJS_2(\rho_0,\rho_1)$ is at least $\alpha(n)$ or at most $\beta(n)$; 
    \item $\measQTDP[\alpha,\beta]$: Decide whether $\measQTD(\rho_0,\rho_1)$ is at least $\alpha(n)$ or at most $\beta(n)$; 
    \item $\QTDP[\alpha,\beta]$: Decide whether $\QTD(\rho_0,\rho_1)$ is at least $\alpha(n)$ or at most $\beta(n)$. 
\end{itemize}
Here, the definitions of $\measQTD$ and $\QTD$ are provided in \Cref{sec:triangular-discrimination}. 

Similar to the polarization lemma for the statistical distance in~\cite{SV97}, Watrous established a polarization lemma for the trace distance~\cite{Wat02}, implying that $\QSDP[\alpha,\beta]$ is in $\QSZK$ for the constant polarizing regime -- namely, constants $\alpha$ and $\beta$ such that $0 \leq \beta < \alpha^2 \leq 1$. 
This work further stated that $\QSDP[\alpha,\beta]$ with any constants $\alpha$ and $\beta$ in this parameter regime is \QSZK{}-complete.\footnote{We do not distinguish \QSZK{} from the honest-verifier variant $\QSZK_{\sf HV}$, since they are equivalent~\cite{Wat09}.} 
In addition, as stated in~\cite[Theorem 3.14]{BDRV19}, the polarization lemma for the statistical distance also implies an improved \SZK{}-hardness for \SDP{}. Consequently, we derive the counterpart improved \QSZK{}-hardness for \QSDP{} (\Cref{thm:QSD-is-QSZKhard}) and omit the detailed proof: 
\begin{theorem}[Improved \QSZK{}-hardness for \QSDP{}]
\label{thm:QSD-is-QSZKhard}
Let $\alpha(n)$ and $\beta(n)$ be efficiently computable functions satisfying $\alpha^2(n)-\beta(n) \geq 1/O(\log{n})$. For any constant $\epsilon \in (0,1/2)$, $\QSDP[\alpha,\beta]$ is \QSZK{}-hard when $\alpha(n) \leq 1-2^{-n^{1/2-\epsilon}}$ and $\beta(n) \geq 2^{-n^{1/2-\epsilon}}$ for every $n\in\bbN$.
\end{theorem}

Furthermore, let \coQSDP{} denote the complement of \QSDP{}. Noting that \QSZK{} is closed under the complement~\cite{Wat02,Wat09}, \coQSDP{} is thus also \QSZK{}-complete. 

\subsubsection{Quantum entropy difference problem}
\label{subsubsec:QEDP}
The definition of the \textsc{Quantum Entropy Difference Problem}, denoted as $\QEDP[g]$, slightly differs from the flavor of \Cref{def:QSDP}: 
\begin{definition}[Quantum Entropy Difference Problem, ${\QEDP[g]}$, adapted from~{\cite[Section 1.2]{BASTS10}}]
	\label{def:QEDP}
    Let $Q_0$ and $Q_1$ be quantum circuits that act on $m(n)$ qubits and having $n$ specified output qubits, where $m$ is polynomial in $n$. For $i \in \binset$, let $\rho_i$ be the state obtained by running $Q_i$ in $\ket{0^n}$ and tracing out the non-output qubits. Let $g:\bbN\rightarrow\bbR^+$ be an efficiently computable function. Then promise that one of the following cases will occur: 
\begin{itemize}[topsep=0.33em, itemsep=0.33em, parsep=0.33em]
	\item \emph{Yes:} A pair of quantum circuits $(Q_0,Q_1)$ such that $\S(\rho_0)-\S(\rho_1) \geq g(n)$;
	\item \emph{No:} A pair of quantum circuits $(Q_0,Q_1)$ such that $\S(\rho_1)-\S(\rho_0) \geq g(n)$.
\end{itemize}
\end{definition}

As implicitly demonstrated in~\cite{BASTS10}, the \QSZK{} containment of $\QEDP[g]$ holds even when $g(n)$ is polynomially small: 
\begin{theorem}[Implicitly in~\cite{BASTS10}]
	\label{thm:QEDP-is-SZK-complete}
	For any efficiently computable function $g(n)$ satisfying $g(n) \geq 1/\poly(n)$, it holds that $\QEDP[g(n)]$ is in $\QSZK$. 
\end{theorem}

\begin{proof}
    It suffices to show a promise gap amplification that reduces $\QEDP[g]$ to $\QEDP[1/2]$.
    Consider new states $\trho_0$ and $\trho_1$ where $\trho_k=\rho_k^{\otimes p(n)}$ for $k\in\binset$ and $p(n)$ is a polynomial of $n$ such that $p(n)g(n)\geq 1/2$. Noting that von Neumann entropy is additive for independent systems, for \textit{yes} instances, we obtain that $\S\!\left(\trho_0\right) - \S\!\left(\trho_1\right) 
    = p(n) \cdot \big( \S(\rho_0)-\S(\rho_1) \big) \geq p(n)g(n) \geq 1/2$. 
    Likewise, we deduce that $\S(\trho_1)-\S(\trho_0) \geq 1/2$ for \textit{no} instances, as desired. 
\end{proof}

%%%%%%%%%%%%%%%%%%%%%%%%%%%%%%%%%%%%%%%%%%%%%%%%%%%%%

\section{Quantum analogs of the triangular discrimination}
\label{sec:triangular-discrimination}

In this section, we introduce \textit{two quantum analogs of the triangular discrimination} and demonstrate their relationships with several commonly used distances, such as trace distance, Bures distance (closely related to the fidelity), and quantum Jensen-Shannon divergence. 

To the best of our knowledge, there is no known quantum analog of triangular discrimination (also known as Vincent-Le Cam divergence). Since triangular discrimination is a symmetrized version of $\chi^2$ divergence, $\TD(p_0,p_1)=\chi^2(p_0\big\|\frac{p_0+p_1}{2})=\chi^2(p_1\big\|\frac{p_0+p_1}{2})$, we present the first quantum analog which is derived from the quantum $\chi^2$ divergence in~\cite{TKRWV10}. 

\begin{definition}[Quantum Triangular Discrimination]
    \label{def:QTD}
    The quantum triangular discrimination between two quantum states $\rho_0$ and $\rho_1$ is defined as 
    \[\QTD(\rho_0,\rho_1)  \coloneqq  \frac{1}{2}\Tr\left( (\rho_0-\rho_1)(\rho_0+\rho_1)^{-1/2}(\rho_0-\rho_1)(\rho_0+\rho_1)^{-1/2} \right).\]
    
    \noindent Furthermore, if $\rho_0+\rho_1$ is not full-rank, then the inverse is defined only on its support. 
\end{definition}

It is noteworthy that this quantum analog of triangular discrimination can be generally defined as $\QTD_{\alpha}(\rho_0,\rho_1)=\chi^2_{\alpha}\left(\rho_z\big\|\frac{\rho_0+\rho_1}{2}\right)$ for $z\in\binset$, following the approach presented in~\cite{TKRWV10}. However, $\QTD_{\alpha}$ is only upper-bounded by the trace distance for $\alpha=1/2$.\footnote{See  \Cref{remark:QTD-vs-traceDist} for the details.} Therefore, we use $\QTD_{\alpha=1/2}(\rho_0,\rho_1)$ for defining \QTD{} in this paper. 

In addition,  we establish another quantum analog of triangular discrimination, denoted by the \textit{Measured Quantum Triangular Discrimination} ($\measQTD$), based on distributions induced by quantum measurements in terms of \Cref{eq:measured-f-divergences}. By utilizing~\cite[Lemma 5]{TV15}, we can derive an explicit formula for $\measQTD$.\footnote{Given $\TD(p_0,p_1)=\chi^2(p_z\big\|\frac{p_0+p_1}{2})$ for $z\in\binset$, an explicit formula for $\measQTD$ follows~\cite[Lemma 5]{TV15}: $\measQTD(\rho_0,\rho_1)=\Tr\big(\frac{\rho_0-\rho_1}{2} \Omega_{\rho_{+}}\!\big( \frac{\rho_0-\rho_1}{2} \big)\big)$ where  $\rho_{+}\! \coloneqq \!\frac{\rho_0+\rho_1}{2}$ and the linear operator $\Omega_{\rho}$ satisfies $\Omega_{\rho}^{-1}(A)=(\rho A+A\rho)/2$. In particular, following the observation in \cite[Section 3.1.2]{BOW19}, if $\rho_{+}=(\beta_1,\cdots,\beta_d)$ is diagonal of full rank, then $\measQTD(\rho_0,\rho_1) = \sum_{i,j=1}^d \frac{2}{\beta_i+\beta_j} |(\rho_{-})_{ij}|^2$ where $\rho_{-}  \coloneqq  \frac{\rho_0-\rho_1}{2}$. \label{footnote:measQTD-explicit-formula}}
As is typical, \QTD{} is lower-bounded by its measured variant $\measQTD$, following from a data-processing inequality for the quantum $\chi^2$-divergence~\cite[Proposition 6]{TKRWV10}: 

\begin{proposition}
\label{prop:measQTD-leq-QTD}
For any quantum states $\rho_0$ and $\rho_1$, $\QTD(\rho_0,\rho_1) \geq \QTD^{\rm meas}(\rho_0,\rho_1)$. 
\end{proposition}

\begin{proof}
According to~\cite[Proposition 6]{TKRWV10}, a data-processing inequality for the quantum $\chi^2_{\alpha=1/2}$-divergence, we have: for any quantum states $\rho_0$ and $\rho_1$, 
\begin{align*}
    \QTD(\rho_0,\rho_1) &=\chi^2_{\alpha=1/2}\rbra*{\rho_0 \bigg\| \frac{\rho_0+\rho_1}{2}}\\
    &\geq \chi^2_{\alpha=1/2}\rbra*{\calM(\rho_0) \bigg\| \calM(\frac{\rho_0+\rho_1}{2})} \\
    &=\tilde{\chi}^2_{\alpha=1/2}\rbra*{\rho_0 \bigg\| \frac{\rho_0+\rho_1}{2}} \\
    &=\QTD^{\rm meas}(\rho_0,\rho_1)
\end{align*}
Here, we denote the measured $\chi^2$-divergence as $\tilde{\chi}^2_{\alpha}(\cdot,\cdot)$, which is defined in terms of \Cref{eq:measured-f-divergences}. Additionally, we choose the quantum channel $\calM$ that corresponds to the optimal POVM in $\tilde{\chi}^2_{\alpha}\left(\rho_0 \big\| \frac{\rho_0+\rho_1}{2}\right)$. 
\end{proof}

\vspace{1em}
We now present three theorems that examine the relationships between the quantum triangular discrimination (\QTD{}) and other commonly used quantum distances and divergences. 
\Cref{thm:QTD-vs-td} compares \QTD{} with the trace distance ($\td$) and is established through a combination of \Cref{lemma:QTD-leq-QSD,lemma:QSDsquare-leq-QTD} in \Cref{subsec:QTD-vs-td}. The latter relies on the trace distance being also a measured version of the statistical distance. 
\begin{theorem}[\QTD{} vs.\ trace distance]
\label{thm:QTD-vs-td}
For any quantum states $\rho_0$ and $\rho_1$, it holds that
\[\td^2(\rho_0,\rho_1) \leq \measQTD(\rho_0,\rho_1) \leq \QTD(\rho_0,\rho_1) \leq \td(\rho_0,\rho_1).\]
\end{theorem}

\Cref{thm:QTD-vs-QJS} demonstrates the relationship between \QTD{} and the quantum Jensen-Shannon divergence (\QJS{}), which is based on a combination of \Cref{lemma:QTDsquare-leq-QJS,lemma:QJS-leq-QTD} in \Cref{subsec:QTD-vs-Bures}. The proof of these lemmas takes advantage of inequalities on the trace distance, thereby linking \QJS{} and \QTD{}. 
\begin{theorem}[\QTD{} vs.\ \QJS{}]
\label{thm:QTD-vs-QJS}
For any quantum states $\rho_0$ and $\rho_1$, it holds that
\[\frac{1}{2}\QTD^2(\rho_0,\rho_1) \leq \QJS(\rho_0,\rho_1) \leq \QTD(\rho_0,\rho_1).\]
\end{theorem}

\Cref{thm:QTD-vs-B} explores the relationship between the \QTD{} and the Bures distance. The bounds of $\measQTD$ (\Cref{lemma:B-leq-measQTD-leq-2B}) rely on the Bures distance being the measured version of the Hellinger distance, while the bounds of \QTD{} (\Cref{prop:QTD-leq-B}) are established using inequalities involving the trace distance. The detailed proof can be found in \Cref{subsec:QTD-vs-QJS}. 
\begin{theorem}[\QTD{} vs.\ Bures distance]
\label{thm:QTD-vs-B}
For any quantum states $\rho_0$ and $\rho_1$, it holds that
\[\frac{1}{2}\Bsquare(\rho_0,\rho_1) \leq \measQTD(\rho_0,\rho_1) \leq \Bsquare(\rho_0,\rho_1) \text{ and } \frac{1}{2} \Bsquare(\rho_0,\rho_1) \leq \QTD(\rho_0,\rho_1) \leq \B(\rho_0,\rho_1).\]
\end{theorem}

\subsection{\QTD{} vs.\ trace distance}
\label{subsec:QTD-vs-td}
We begin by establishing the challenging direction (in \Cref{thm:QTD-vs-td}) that \QTD{} is upper-bounded by the trace distance (\Cref{lemma:QTD-leq-QSD}), as well as highlighting two important subtleties of \QTD{}. The proof of the converse direction will be provided at the end of this subsection. 

\begin{lemma}[$\QTD \leq \td$]
\label{lemma:QTD-leq-QSD}
For any quantum states $\rho_0$ and $\rho_1$, $\QTD(\rho_0,\rho_1) \leq \td(\rho_0,\rho_1).$
\end{lemma}

The first subtlety of \QTD{} lies in the fact that the inequality in \Cref{lemma:QTD-leq-QSD} holds solely for a particular choice of $\alpha=1/2$  for $\QTD_{\alpha}$ (leading to the minimum):

\begin{remark}[$\QTD_{\alpha} \leq \td$ holds only for $\alpha=1/2$]
\label{remark:QTD-vs-traceDist}
As~\cite[Proposition 7]{TKRWV10} implies that $\QTD_{\alpha=1/2} \leq \QTD_{\alpha}$, we may wonder whether \Cref{lemma:QTD-leq-QSD} holds for any $\alpha\in[0,1]$. Here is a counterexample: Consider two single-qubit pure quantum states  $\rho^*_0=\frac{1}{2}\left(I+\frac{6}{7}\sigma_X+\frac{3}{7}\sigma_Y+\frac{2}{7}\sigma_Z\right)$ and $\rho^*_1=\frac{1}{2}\left(I-\frac{3}{7}\sigma_X-\frac{2}{7}\sigma_Y+\frac{6}{7}\sigma_Z\right)$, where $\sigma_X$, $\sigma_Y$ and $\sigma_Z$ are Pauli matrices.
Then we simply have $\QTD_{\alpha=1/2}(\rho^*_0,\rho^*_1) = \td(\rho^*_0,\rho^*_1) < \QTD_{\alpha>1/2}(\rho^*_0,\rho^*_1)$.
\end{remark}

The second subtlety of \QTD{} concerns the notable difference in the equality condition of this inequality (\Cref{prop:conds-QTD-eq-QSD}) compared to its classical counterpart. Specifically, the classical counterpart merely requires \Cref{prop:conds-QTD-eq-QSD}\ref{thmitem:conds-QTD-eq-QSD-classical}.\footnote{In particular, $(p_0(x)-p_1(x))^2=(p_0(x)+p_1(x))^2$ holds for any $x\in \supp{p_0}\cup\supp{p_1}$.} Nevertheless, the inequalities in \Cref{thm:QTD-vs-td} exhibit a similar behavior to the inequalities between the corresponding classical distances, namely triangular discrimination (TD) and statistical difference (SD).

\begin{proposition}[Equality condition for $\QTD\leq\td$]
\label{prop:conds-QTD-eq-QSD}
For any quantum states $\rho_0$ and $\rho_1$, the equality $\QTD(\rho_0,\rho_1)=\td(\rho_0,\rho_1)$ holds if and only if these states satisfy the following conditions: 
\begin{enumerate}[label={\upshape(\roman*)}, topsep=0.33em, itemsep=0.33em, parsep=0.33em]
    \item \label{thmitem:conds-QTD-eq-QSD-classical} $(\rho_0-\rho_1)(\rho_0+\rho_1)^{-1}(\rho_0-\rho_1)=(\rho_0+\rho_1)$; 
    \item $(\rho_0-\rho_1)^{\dagger}(\rho_0-\rho_1)=\frac{\Tr\left[(\rho_0-\rho_1)^{\dagger}(\rho_0-\rho_1)\right]}{|\supp{\rho_0-\rho_1}|}I$;
    \item For any $k\in\supp{\rho_0-\rho_1}$, 
    \[\sign\rbra*{\lambda_k(\rho_0-\rho_1)}=\sign\rbra*{\lambda_k\big( (\rho_0+\rho_1)^{-1/2}(\rho_0-\rho_1)(\rho_0+\rho_1)^{1/2} \big)},\] where $\lambda_k(A)$ is the $k$-th eigenvalue of the matrix $A$. 
\end{enumerate}
\end{proposition}

\vspace{1em}

We now outline the proof of \Cref{lemma:QTD-leq-QSD}: Firstly, we establish an upper bound of \QTD{} by the trace distance with an infinite norm (multiplicative) factor using a matrix version of H{\"o}lder inequality. Subsequently, we bound this infinite norm factor by analyzing its largest singular value employing the Weyl's inequalities. The detailed proof follows below. 

\begin{proof}[Proof of \Cref{lemma:QTD-leq-QSD}]
By utilizing a matrix H{\"o}lder inequality, such as~\cite[Corollary IV.2.6]{Bhatia96}, we obtain
\begin{equation}
\begin{aligned}
\label{eq:QTD-leq-QSD-Holder}
\QTD(\rho_0, \rho_1) 
&= \frac{1}{2} \Tr\left( (\rho_0-\rho_1)(\rho_0+\rho_1)^{-1/2}(\rho_0-\rho_1)(\rho_0+\rho_1)^{-1/2} \right) \\
&\leq \frac{1}{2} \|\rho_0-\rho_1\|_1 \cdot \|(\rho_0+\rho_1)^{-1/2} (\rho_0-\rho_1) (\rho_0+\rho_1)^{-1/2}\|_{\infty}\\
\end{aligned}
\end{equation}

It is sufficient to show that 
\[\|(\rho_0+\rho_1)^{-1/2} (\rho_0-\rho_1) (\rho_0+\rho_1)^{-1/2}\|_{\infty} = \sigma_{\max}\left((\rho_0+\rho_1)^{-1/2} (\rho_0-\rho_1) (\rho_0+\rho_1)^{-1/2}\right) \leq 1,\]
where $\sigma_{\max}(A)$ is the largest singular value of $A$. 
Let $\rho \coloneqq \frac{1}{2}(\rho_0+\rho_1)$, then we have
$(\rho_0+\rho_1)^{-1/2} (\rho_0-\rho_1) (\rho_0+\rho_1)^{-1/2} = \rho^{-1/2} (\rho-\rho_1) \rho^{-1/2} = I -  \rho^{-1/2}\rho_1 \rho^{-1/2}.$
Noting that $\rho^{-1/2} \rho_1 \rho^{-1/2}$ is positive semi-definite, and $I-\rho^{-1/2}\rho_1\rho^{-1/2}$ thus is Hermitian.
We then obtain that $|I-\rho^{-1/2}\rho_1\rho^{-1/2}| \preceq I$.\footnote{It suffices to show that $-I \preceq I-\rho^{-1/2}\rho_1\rho^{-1/2} \preceq I$. The right-hand side is evident, while the left-hand side follows from $\rho^{-1/2}\rho_1\rho^{-1/2} \preceq 2I$, which holds by applying $\Phi(\sigma) \coloneqq \rho^{1/2} \sigma \rho^{1/2}$ on both sides.} With the help of~\cite[Corollary 4.3.12]{HJ12}, a corollary of Weyl's inequalities, this inequality implies that:
\begin{equation}
\label{eq:QTD-leq-QSD-eigenvals}
\begin{aligned}
\sigma_{\max}\left(I-\rho^{-1/2}\rho_1\rho^{-1/2}\right) 
&= \lambda_{\max} \left(I-\rho^{-1/2}\rho_1\rho^{-1/2}\right) \\
&\leq \lambda_{\max}\left( \left(I-\rho^{-1/2}\rho_1\rho^{-1/2}\right) + \rho^{-1/2}\rho_1\rho^{-1/2} \right)\\
&\leq 1.
\end{aligned}
\end{equation}

Here, the first line is derived from the fact that the singular values of a Hermitian matrix are equal to the absolute values of the corresponding eigenvalues of the same matrix, and the last line is due to $\lambda_{\max}(I) = 1$.  
\end{proof}

To derive the equality condition of \Cref{lemma:QTD-leq-QSD}, and thereby prove \Cref{prop:conds-QTD-eq-QSD}, a thorough analysis of the equality condition of the matrix H{\"o}lder inequality in~\cite{Ciosmak21} is required. The detailed proof is provided subsequently.

\begin{proof}[Proof of \Cref{prop:conds-QTD-eq-QSD}]
We begin with the equality condition for the matrix H{\"o}lder inequality in~\cite[Theorem 2.11]{Ciosmak21}. Let $A=\frac{\rho_0-\rho_1}{2}$ and $B=\left(\frac{\rho_0+\rho_1}{2}\right)^{-1/2} \left(\frac{\rho_0-\rho_1}{2}\right) \left(\frac{\rho_0+\rho_1}{2}\right)^{-1/2}$, then
\begin{equation}
    \label{eq:Holder-QTD-eq-QSD}
    \frac{A^{\dagger} B}{\Tr|A| \|B\|_{\infty}} = \frac{B^{\dagger} A}{\Tr|A| \|B\|_{\infty}}
    = \frac{|A|}{\Tr|A|} = \frac{|B|^{\infty}}{\Tr\left(|B|^{\infty}\right)}.
\end{equation}
Moreover, $B^{\dagger}A$ is supposed to be symmetric and positive semi-definite. 
Noting that $A$ and $B$ are Hermitian, we obtain $[A,B]=0$ by using the first equality in \Cref{eq:Holder-QTD-eq-QSD}. This equality implies that $B^{\dagger}A$ is indeed symmetric, as well as the singular value decomposition $A=\sum_k \sigma_k(A)\ket{v_k}\bra{v_k}$ and $B=\sum_k \sigma_k(B)\ket{v_k}\bra{v_k}$. Then by \Cref{eq:Holder-QTD-eq-QSD}, we obtain
\begin{align*}
B^{\dagger}A = \sum_k \sigma_k(B)\sigma_{k}(A) \ket{v_{k}}\bra{v_{k}}&=\sigma_{\max}(B)\sum_k\sigma_k(A)\ket{v_k}\bra{v_k}=\|B\|_{\infty}|A|, \\
\frac{|A|}{\Tr|A|}=\sum_k \frac{\sigma_k(A)}{\sum_i\sigma_i(A)} \ket{v_i}\bra{v_i} &= \sum_k \frac{\sigma_k^{\infty}}{\sum_j \sigma_k^{\infty}(B)} \ket{v_k}\bra{v_k}=\frac{|B|^{\infty}}{\Tr\left(|B|^{\infty}\right)}.
\end{align*}
Noting that $\{\ket{v_i}\}_{v_i\in\supp{\rho_0-\rho_1}}$ is an orthonormal basis, by comparing the coefficients, we have
\begin{equation}
    \forall k:~\sigma_k(A) = \sigma_{\max}(A) \text{ and } \sigma_k(B) = \sigma_{\max}(B) = 1. 
\end{equation}
Here, $\sigma_{\max}(B)=1$ due to \Cref{eq:QTD-leq-QSD-eigenvals} with the equality. 
Therefore, we obtain that $B$ is an orthogonal matrix, which is equivalent to $(\rho_0-\rho_1)(\rho_0+\rho_1)^{-1}(\rho_0-\rho_1)=(\rho_0+\rho_1)$. Furthermore, noting that $\Tr(A^{\dagger} A)=\sum_k \sigma^2_k(A)$, this identity implies that $A^{\dagger}A=\frac{\Tr\left(A^{\dagger} A\right)}{|\supp{\rho_0-\rho_1}|}I$ as desired. Finally, to make $B^{\dagger} A$ to be positive semi-definite, we require that  $\sign\rbra*{\lambda_k(A)}=\sign\rbra*{\lambda_k(B)}$ for any $k\in\supp{\rho_0-\rho_1}$, which finishes the proof. 
\end{proof}

\vspace{1em}
Lastly, we present the proof of  \Cref{lemma:QSDsquare-leq-QTD} (the converse direction in \Cref{thm:QTD-vs-td}). In particular, by leveraging \Cref{prop:measQTD-leq-QTD}, we can derive a lower bound for the quantum counterparts of triangular discrimination in terms of the trace distance. 

% The following lemma holds for any $\alpha \in [0,1].
\begin{lemma}[$\td^2 \leq \QTD$]
\label{lemma:QSDsquare-leq-QTD}
For any quantum states $\rho_0$ and $\rho_1$, 
\[\td(\rho_0,\rho_1)^2 \leq \measQTD(\rho_0,\rho_1) \leq \QTD(\rho_0,\rho_1).\]
\end{lemma}

\begin{proof}
Owing to \Cref{prop:measQTD-leq-QTD}, it suffices to show that $\QTD^{\rm meas}(\rho_0,\rho_1) \geq \td(\rho_0,\rho_1)$. Analogous to the approach presented in~\cite{Top00}, we obtain the following for any POVM $\calE$:
\begin{align}
\QTD^{\rm meas}(\rho_0, \rho_1) &\geq \TD\left( p_0^{(\calE)}, p_1^{(\calE)} \right) \notag \\
&= \frac{1}{2} \sum_{x} \frac{\left(p_0^{(\calE)}(x)-p_1^{(\calE)}(x)\right)^2}{p_0^{(\calE)}(x)+p_1^{(\calE)}(x)} \notag \\
&= \sum_x \frac{p_0^{(\calE)}(x)+p_1^{(\calE)}(x)}{2} \cdot \abs*{\frac{p_0^{(\calE)}(x)-p_1^{(\calE)}(x)}{p_0^{(\calE)}(x)+p_1^{(\calE)}(x)}}^2 \notag \\
&\geq \rbra*{ \sum_x \frac{p_0^{(\calE)}(x)+p_1^{(\calE)}(x)}{2} \cdot \abs*{\frac{p_0^{(\calE)}(x)-p_1^{(\calE)}(x)}{p_0^{(\calE)}(x)+p_1^{(\calE)}(x)}} }^2 \notag \\
&=\rbra*{\frac{1}{2}\sum_x \abs*{p_0^{(\calE)}(x)-p_1^{(\calE)}(x)}}^2, \label{eq:QSDsquare-leq-QTD}
\end{align}
where the fourth line is because of $\bbE[X^2]\geq (\bbE[X])^2$ for any random variable $X$. 
We then complete the proof by choosing $\calE$ that maximizes the line of \Cref{eq:QSDsquare-leq-QTD}. 
\end{proof}

\subsection{\QTD{} vs.~(squared) Bures distance}
\label{subsec:QTD-vs-Bures}

We now present inequalities concerning two different quantum analogs of the triangular discrimination (TD), namely \QTD{} and the measured version $\measQTD$, expressed in terms of the Bures distance. Interestingly, these inequalities exhibit divergent behaviors for \QTD{} (\Cref{lemma:B-leq-measQTD-leq-2B}) and $\measQTD$ (\Cref{prop:QTD-leq-B}), and we can identify an example (in \Cref{remark:measQTD-vs-QTD}) that distinguishes between these two quantum analogs of TD. 
These divergent behaviors have implications in quantum complexity theory, particularly in the corresponding polarization lemma and the complexity class \QSZK{}.\footnote{See \Cref{lemma:measQTD-polarization,lemma:QTD-polarization} in \Cref{subsec:polarization-lemma-QTD} for further details.}

\vspace{1em}
We begin by establishing the inequalities between $\measQTD$ and the Bures distance, as stated in \Cref{lemma:B-leq-measQTD-leq-2B}. The proof crucially relies on the fact that the Bures distance corresponds to the measured version of Hellinger distance~\cite{FC94}. 

\begin{lemma}
\label{lemma:B-leq-measQTD-leq-2B}
For any quantum states $\rho_0$ and $\rho_1$, $\frac{1}{2}\Bsquare(\rho_0,\rho_1) \leq \measQTD(\rho_0,\rho_1) \leq \Bsquare(\rho_0,\rho_1)$. 
\end{lemma}

\begin{proof} 
    Let $\calE^*$ be the optimal measurement that attains the value of $\measQTD(\rho_0,\rho_1)$. We first notice that 
    \begin{align*}
        \measQTD(\rho_0,\rho_1) &= \frac{1}{2} \sum_x \frac{\left( p_0^{(\calE^*)}(x) - p_1^{(\calE^*)}(x) \right)^2}{p_0^{(\calE^*)}(x) + p_1^{(\calE^*)}(x)}\\
        &= \frac{1}{2} \sum_x \tfrac{\left(\sqrt{p_0^{(\calE^*)}(x)}-\sqrt{p_1^{(\calE^*)}(x)}\right)^2 \left(\sqrt{p_0^{(\calE^*)}(x)}+\sqrt{p_1^{(\calE^*)}(x)}\right)^2}{p_0^{(\calE^*)}(x) + p_1^{(\calE^*)}(x)}.
    \end{align*}

    Noting that $a^2+b^2 \leq (a+b)^2 \leq 2(a^2+b^2)$ for $a,b\geq 0$, we have derived that
    \begin{equation}
    \label{eq:measQTD}
    \frac{1}{2} \sum_x \rbra*{ \sqrt{p_0^{(\calE)}(x)}-\sqrt{p_1^{(\calE)}(x)} }^2 \leq \measQTD(\rho_0,\rho_1) \leq \sum_x \rbra*{ \sqrt{p_0^{(\calE^*)}(x)}-\sqrt{p_1^{(\calE^*)}(x)} }^2. 
    \end{equation}

    Noting that the Bures distance is the measured Hellinger distance~\cite{FC94}, we have: 
    \begin{itemize}
        \item For the lower bound, since the first inequality in \Cref{eq:measQTD} holds for arbitrary POVM $\calE$, we choose $\calE'$ that maximizes the measured Hellinger distance, which indicates: 
    \[\measQTD(\rho_0,\rho_1) \geq  \frac{1}{2} \sum_x \left(\sqrt{p_0^{(\calE')}(x)}\!-\!\sqrt{p_1^{(\calE')}(x)}\right)^2 = \sup\limits_{{\rm POVM}~\calE} \!\Hsquare\left( p_0^{(\calE)}, p_1^{(\calE)} \right) = \frac{1}{2}\Bsquare(\rho_0,\rho_1).\]
        \item For the upper bound, let $\calE'$ be the POVM measurement that maximizes the measured Hellinger distance. By the second inequality in \Cref{eq:measQTD}, we deduce: 
    \[\measQTD(\rho_0,\rho_1) \leq \sum_x \left(\sqrt{p_0^{(\calE')}(x)}\!-\!\sqrt{p_1^{(\calE')}(x)}\right)^2 = \sup\limits_{{\rm POVM}~\calE} \!2\Hsquare\left( p_0^{(\calE)}, p_1^{(\calE)} \right) = \Bsquare(\rho_0,\rho_1). \qedhere\]
    \end{itemize}
\end{proof}

Next, we present the inequalities between \QTD{} and the Bures distance, as detailed in \Cref{prop:QTD-leq-B}. It is noteworthy that the upper bound in these inequalities is as weak as the trace distance, and we further provide an example (in \Cref{remark:measQTD-vs-QTD}) to distinguish these two quantum analogs of TD in terms of the Bures distance.

\begin{proposition}
\label{prop:QTD-leq-B}
For any quantum states $\rho_0$ and $\rho_1$, $\frac{1}{2}\Bsquare(\rho_0,\rho_1) \leq \QTD(\rho_0,\rho_1) \leq \B(\rho_0,\rho_1)$. 
\end{proposition}
\begin{proof}
We establish the left-hand side inequality by plugging \Cref{prop:measQTD-leq-QTD} into \Cref{lemma:B-leq-measQTD-leq-2B}. The right-hand side inequality follows from combining \Cref{prop:traceDist-vs-B} and \Cref{lemma:QTD-leq-QSD}.
\end{proof}

\begin{remark}[$\measQTD$ vs.~\QTD{}]
\label{remark:measQTD-vs-QTD}
The squared Bures distance is an example that separates between $\measQTD$ and \QTD{}: 
Utilizing the counterexample $\rho_0^*$ and $\rho_1^*$ defined in \Cref{remark:QTD-vs-traceDist}, we can obtain $\measQTD(\rho^*_0,\rho^*_1) \leq \Bsquare(\rho^*_0,\rho^*_1) < \QTD_{\alpha=1/2}(\rho^*_0,\rho^*_1) = \td(\rho^*_0,\rho^*_1) < \B(\rho^*_0,\rho^*_1).$
\end{remark}

\subsection{\QTD{} vs.~\QJS{}}
\label{subsec:QTD-vs-QJS}
We now establish the inequalities between \QTD{} and \QJS{}. It is worth noting that the corresponding classical distance, the triangular discrimination (TD), serves as a constant multiplicative-error approximation of the Jensen-Shannon divergence (JS), as illustrated by the inequalities $\frac{1}{2}\TD(p_0,p_1) \leq \JS(p_0,p_1) \leq \ln{2} \cdot \TD(p_0,p_1)$ in~\cite[Theorem 2]{Top00}. However, such a property does not extend to \QTD{} and \QJS{}.\footnote{For further details, please refer to \Cref{footnote:QTD-vs-QJS-tight-bound}.}

\vspace{1em}
We start by establishing the lower bound of \QJS{} in terms of \QTD{}, as stated in \Cref{lemma:QTDsquare-leq-QJS}. The proof straightforwardly follows from inequalities concerning the trace distance.

\begin{lemma}
\label{lemma:QTDsquare-leq-QJS}
For any quantum states $\rho_0$ and $\rho_1$, 
$\frac{1}{2}\QTD^2(\rho_0,\rho_1) \leq \QJS(\rho_0,\rho_1)$.
\end{lemma}

\begin{proof}
Combining \Cref{lemma:QTD-leq-QSD,lemma:traceDist-leq-QJS}, we obtain that: for any quantum states $\rho_0$ and $\rho_1$,
\[\QJS(\rho_0,\rho_1) \geq \sum_{v=1}^{\infty} \frac{\td(\rho_0,\rho_1)^{2v}}{2v(2v-1)} \geq \sum_{v=1}^{\infty} \frac{\QTD(\rho_0,\rho_1)^{2v}}{2v(2v-1)} \geq \frac{1}{2} \QTD^2(\rho_0,\rho_1),\]
where the last inequality uses the first-order approximation. This completes the proof. 
\end{proof}

Next, we present the upper bound of \QJS{} in terms of \QTD{}, as detailed in \Cref{lemma:QJS-leq-QTD}. The proof strategies is analogous to the proof of~\cite[Theorem 8]{TKRWV10}.

\begin{lemma}
\label{lemma:QJS-leq-QTD}
For any quantum states $\rho_0$ and $\rho_1$, 
$\QJS(\rho_0,\rho_1) \leq \QTD(\rho_0,\rho_1)$.
\end{lemma}

\begin{proof}
We begin by noting that an upper bound for the quantum relative entropy in~\cite{RS90}: 
\begin{equation}
    \label{eq:QRS-upper-bound}
    \D(\rho_0 \| \rho_1) \leq \frac{1}{\gamma} \Tr\left(\rho_0^{1+\gamma} \rho_1^{-\gamma}-\rho_0\right) 
    = \frac{1}{\gamma} \left[ \Tr\left(\rho_0^{1+\gamma}\rho_1^{-\gamma}\right) -1 \right] \text{ for } 0 < \gamma \leq 1. 
\end{equation}

Since the quantum Jensen-Shannon divergence is a symmetrized version of the quantum relative entropy, we deduce the following by setting $\gamma = 1/2$ in \Cref{eq:QRS-upper-bound}:
\begin{align*}
\QJS(\rho_0,\rho_1) &= \frac{1}{2} \sum_{z\in\binset} \D\!\left( \rho_z \bigg\| \frac{\rho_0+\rho_1}{2} \right)\\
&\leq \sum_{z\in\binset} \left[\Tr\!\left(\rho_z^{3/2} \left(\frac{\rho_0+\rho_1}{2}\right)^{-1/2}\right)-1\right]\\
&\leq \frac{1}{2}\sum_{z\in\binset} \left[ \Tr\!\left( \rho_z \left(\frac{\rho_0+\rho_1}{2}\right)^{-1/2} \rho_z \left(\frac{\rho_0+\rho_1}{2}\right)^{-1/2} \right) -1 \right]\\
&= \QTD(\rho_0,\rho_1),
\end{align*}
where the third line follows from $\Tr\!\left[ \left(\rho_z^{1/2}\rho^{-1/2}\rho_z^{1/2} -\rho_z^{1/2} \right)^{\dagger} \left(\rho_z^{1/2}\rho^{-1/2}\rho_z^{1/2} -\rho_z^{1/2} \right) \right] \geq 0$, since $\rho_z^{1/2}\rho^{-1/2}\rho_z^{1/2}$ is positive semi-definite and thus $\rho_z^{1/2}\rho^{-1/2}\rho_z^{1/2} -\rho_z^{1/2}$ is Hermitian. 
\end{proof}

%%%%%%%%%%%%%%%%%%%%%%%%%%%%%%%%%%%%%%%%%%%%%%%%%%%%%

\section{Complete problems for \QSZK{} on the quantum state testing}
\label{sec:new-QSZK-complete-problems}

In this section, we introduce three new \QSZK{} complete problems: the \textsc{Quantum Jensen-Shannon Divergence Problem} (\QJSP{}), the \textsc{Measured Quantum Triangular Discrimination Problem} (\measQTDP{}), and the \textsc{Quantum Triangular Discrimination Problem} (\QTDP{}). The promise problems \QJSP{} and \measQTDP{} establish the \textit{proper} quantum analogs of the classical problems investigated in~\cite{BDRV19} and exhibit how their behavior differs from the classical counterparts. 

\begin{theorem}[\QJSP{} is \QSZK{}-complete]
\label{thm:QJSP-is-QSZK-complete}
Let $\alpha(n)$ and $\beta(n)$ be efficiently computable functions such that $0 \leq \beta < \alpha \leq 1$, where $n$ denotes the number of qubits used by quantum states $\rho_0$ and $\rho_1$. Then, the following holds\emph{:}
\[ \text{For any } \alpha(n)-\beta(n) \geq 1/\poly(n), \QJSP[\alpha,\beta] \text{ is in } \QSZK. \]
Furthermore, $\QJSP[\alpha,\beta]$ is \QSZK{}-hard if $\alpha(n) \leq 1-2^{-n^{1/2-\epsilon}}$ and $\beta(n) \geq 2^{-n^{1/2-\epsilon}}$ for some constant $\epsilon\in(0,1/2)$ and sufficiently large $n$. 
\end{theorem}

\begin{theorem}[\measQTDP{} and \QTDP{} are \QSZK{}-complete]
\label{thm:measQTDP-is-QSZK-complete}
Let $\alpha(n)$ and $\beta(n)$ be efficiently computable functions such that $0 \leq \beta < \alpha \leq 1$, where $n$ denotes the number of qubits used by quantum states $\rho_0$ and $\rho_1$. Then, it holds that\emph{:}
\begin{enumerate}[label={\upshape(\arabic*)}, topsep=0.33em, itemsep=0.33em, parsep=0.33em]
    \item For any $\alpha(n)-\beta(n) \geq 1/O(\log{n})$, $\measQTDP[\alpha,\beta]$ is in \QSZK{}; 
    \item \label{thmitem:QTDP-in-QSZK} For any $\alpha^2(n)-\beta(n) \geq 1/O(\log{n})$, $\QTDP[\alpha,\beta]$ is in \QSZK{}.
\end{enumerate}
Furthermore, $\measQTDP[\alpha,\beta]$ and $\QTDP[\alpha,\beta]$ are \QSZK{}-hard if $\alpha(n) \leq 1-2^{-n^{1/2-\epsilon}}$ and $\beta(n) \geq 2^{-n^{1/2-\epsilon}}$ for some constant $\epsilon\in(0,1/2)$ and sufficiently large $n$. 
\end{theorem}

Furthermore, by using the reductions used to establish \Cref{thm:QJSP-is-QSZK-complete}, we achieve a simple \QSZK{}-hardness proof for the \textsc{Quantum Entropy Difference Problem} (\QEDP{}) introduced by Ben-Aroya, Schwartz, and Ta-Shma~\cite{BASTS10}: 

\begin{corollary}[Simple \QSZK{}-hardness for \QEDP{}]
    \label{corr:simple-QSZKhardness-QEDP}    
    There exists a constant $\epsilon\in(0,1/2)$ such that 
    $\QEDP[g(n)]$ is \QSZK{}-hard when $g(n) \leq \frac{\ln{2}}{2} \big( 1-2^{(n-3)^{1/2-\epsilon}+1} \big)$ for sufficiently large $n$. 
\end{corollary}

Subsequently, we proceed to demonstrate the proof of these theorems. 

\subsection{\QSZK{} containment using the quantum entropy extraction}

Along the line of~\cite{BDRV19}, we implicitly employ the quantum entropy extraction approach to polarize quantum distances~\cite{BASTS10}. This approach leads to the \QSZK{} containment of \QJSP{}, as stated in \Cref{lemma:QJSP-is-in-QSZK}, with a promise gap that is \textit{inverse-polynomial}. This containment is accomplished through establishing a reduction from \QJSP{} to \QEDP{}. For a concise overview of \QEDP{}, please refer to \Cref{subsubsec:QEDP}. 

\begin{lemma}[\QJSP{} is in \QSZK{}]
    \label{lemma:QJSP-is-in-QSZK} 
    For any $0 \leq \beta(n) < \alpha(n) \leq 1$ such that $\alpha(n)-\beta(n) \geq 1/p(n)$ where $p(n)$ is some polynomial of $n$, we have $\QJSP[\alpha,\beta]$ is in $\QSZK$.
\end{lemma}

With inequalities between the trace distance and $\QJS_2$, we further derive a \QSZK{} containment with an inverse-polynomial promise gap for \QSDP{} on some parameter regime: 
\begin{theorem}
    \label{thm:QSDP-in-QSZK-inversepoly}
    For any $0 \leq \sqrt{2\ln{2}}\beta(n) < \alpha^2(n) \leq 1$ such that $\alpha^2(n)-\sqrt{2\ln{2}}\beta(n) \geq 1/p(n)$ where $p(n)$ is some polynomial of $n$, we know that $\QSDP[\alpha^2,\sqrt{2\ln{2}}\beta]$ is in \QSZK{}. 
\end{theorem}

\begin{proof}
The reduction from \QSDP{} to \QJSP{} directly follows from the inequalities on \QJS{}:
    For \textit{yes} instances, $\QJS_2(\rho_0,\rho_1)\geq \alpha^2$ implies that $\td(\rho_0,\rho_1)\geq \alpha^2$ due to \Cref{lemma:QJS-leq-traceDist}; 
    whereas for \textit{no} instances, $\QJS_2(\rho_0,\rho_1) \leq 2\ln{2}\cdot\beta^2$ implies that $2\ln{2}\cdot\beta^2 \geq \sum_{v=1}^{\infty}\frac{\td(\rho_0,\rho_1)^{2v}}{v(2v-1)} \geq \td(\rho_0,\rho_1)^2$ as desired, where the last inequality utilizes the first-order approximation.
\end{proof}

It is noteworthy that $\TDP[\alpha,\beta]$ is in \SZK{} when $\alpha(n)-\beta(n)$ is at least some inverse polynomial~\cite{BDRV19}. However, we are unlikely to have a similar reduction from \QTDP to \QJSP{} since these distances behave differently from their classical counterpart: 

\begin{remark}[An obstacle to a reduction from \QTDP to \QJSP{}]
The \SZK{} containment of \TDP{} follows from a tailor-made (Karp) reduction from \TDP{} to \JSP{}. The key observation is that $\TD(p_0,p_1)$ is a \textit{constant multiplicative-error approximation} of $\JS_2(p_0,p_1)$, and specifically, the lower bound $\TD(p_0,p_1)/2$ is exactly the first-order approximation of the series used in the upper bound $\ln{2}\cdot\TD(p_0,p_1)$. 
Utilizing this fact, \cite[Lemma 4.5]{BDRV19} establishes that $\lambda^2\TD(p_0,p_1)$ is a $1/\poly(n)$-\textit{additive error approximation} of $\JS_2(q_0,q_1)$, where $\lambda$ is some specific $1/\poly(n)$ factor and $q_0$ (also $q_1$) is a convex combination of $p_0$ and $p_1$ parameterized by $\lambda$.
However, $\QTD(\rho_0,\rho_1)$ is \textit{not} a constant multiplicative error approximation of $\QJS_2(\rho_0,\rho_1)$.\footnote{Numerical simulations suggest that the tight bound is $\QTD^2(\rho_0,\rho_1) \leq \QJS_2(\rho_0,\rho_1) \leq \QTD(\rho_0,\rho_1)$, while we only managed to prove a slightly weaker bound $\frac{1}{2}\QTD^2(\rho_0,\rho_1) \leq \QJS_2(\rho_0,\rho_1) \leq \QTD(\rho_0,\rho_1)$ in \Cref{thm:QTD-vs-QJS}.\label{footnote:QTD-vs-QJS-tight-bound}} 
\end{remark}

\subsection{\QJSP{} is in \QSZK{}} 
For any given \QJSP{} instance and its corresponding states $\rho_0$ and $\rho_1$, the \QSZK{} containment of \QJSP{} essentially follows from an equality concerning $\S(\rho'_0)-\S(\rho'_1)$ and $\QJS(\rho_0,\rho_1)$, where the preparation of $\rho'_0$ and $\rho'_1$ requires additional gadgets using $\rho_0$ and $\rho_1$ as building blocks.
This approach resembles the classical proof from~\cite[Proposition 4.1 and Lemma 4.2]{BDRV19}. 

However, several modifications are required due to discrepancies between classical and quantum probabilities. In particular, the classical proof relies on a probability that conditions on distributions are supposed to be distinguished, whereas its quantum counterpart -- quantum conditional probability -- is not well-defined in general. To address the challenge, we circumvent this issue by instead considering a conditional entropy of classical-quantum states conditioned on a classical register.

\begin{proof}[Proof of \Cref{lemma:QJSP-is-in-QSZK}]
    Since $\QEDP[g]$ is in \QSZK{} for any $g(n)$ that is at least $1/\poly(n)$, as stated in \Cref{thm:QEDP-is-SZK-complete}, the proof is primarily a reduction from $\QJSP[\alpha,\beta]$ to $\QEDP[g]$, where $0 \leq \beta(n) < \alpha(n) \leq 1$ and $\alpha(n)-\beta(n)$ is at least an inverse polynomial of $n$. We will specify the function $g$ later. 
 
    Let $Q_0$ and $Q_1$ be the given quantum circuits acting on $m(n)$ all-zero qubits, which produce $n$-qubits quantum states $\rho_0$ and $\rho_1$, respectively, by tracing out the non-output qubits. 
    
	Now, consider a classical-quantum mixed state on a classical register $\sfB$ and a quantum register $\sfY$, denoted as $\rho'_1 = \frac{1}{2}\ket{0}\bra{0}\otimes \rho_0 + \frac{1}{2} \ket{1}\bra{1}\otimes \rho_1$. 
    We apply our reduction to produce quantum circuits $Q'_0$ and $Q'_1$, which prepare classical-quantum mixed states $\rho'_0$ and $\rho'_1$, respectively.
    In particular, $\rho'_0=(p_0\ket{0}\bra{0}+p_1\ket{1}\bra{1})\otimes(\frac{1}{2}\rho_0+\frac{1}{2}\rho_1)$, and $\sfB'=(p_0,p_1)$ is an independent random bit with $\H(\sfB')=1-\frac{1}{2}[\alpha(n)+\beta(n)]$. 
    
    By using a rotation gate  $R_{\theta}$ such that $R_{\theta} \ket{0} = \sqrt{p_0} \ket{0} + \sqrt{p_1} \ket{1}$, we provide the quantum circuit description of $Q'_0$ and $Q'_1$ in \Cref{fig:QJSP-in-QSZK-Qzero} and \Cref{fig:QJSP-in-QSZK-Qone}, respectively. Here, $\sfA$ and $\sfA'$ are ancillary single-qubit registers, and quantum registers $\sfY$ and $\sfZ$ collectively act on $m(n)$ qubits. 
    \begin{figure}[ht!]
    \centering
    \begin{minipage}[t]{0.48\textwidth}
    \centering
    \begin{quantikz}[column sep = 0.6em, row sep = 0em, inner sep = 0em]
	   \lstick{$\ket{0}_{\sfB'}$} & \gate{H} & \ctrl{1} & \gate{R_{\theta}} & \qw & \qw & \qw\\
	   \lstick{$\ket{0}_{\sfA'}$} & \qw & \targ{} & \qw & \qw & \qw & \trash{\text{trace}} \\
	   \lstick{$\ket{0}_{\sfA}$} & \gate{H} & \ctrl{1} & \gate{X} & \ctrl{1} & \gate{X} & \trash{\text{trace}}\\
	   \lstick{$\ket{\bar{0}}_{\sfY}$} & \qwbundle[]{} & \gate[2]{Q_1} & \qwbundle[]{} & \gate[2]{Q_0} & \qwbundle[]{} & \qw\\
	   \lstick{$\ket{\bar{0}}_{\sfZ}$} & \qwbundle[]{} & & \qwbundle[]{} & & \qwbundle[]{} & \trash{\text{trace}}
    \end{quantikz}
    \caption{Quantum circuit $Q'_0$. }
    \label{fig:QJSP-in-QSZK-Qzero}
    \end{minipage}
    \begin{minipage}[t]{0.48\textwidth}
    \centering
    \begin{quantikz}[column sep = 0.5em, row sep = 0.66em]
	   \lstick{$\ket{0}_{\sfB}$} & \gate{H} & \ctrl{1} & \ctrl{2} & \gate{X} & \ctrl{2} & \gate{X} & \qw\\
	   \lstick{$\ket{0}_{\sfA}$} & \qw & \targ{} & \qw & \qw & \qw & \qw & \trash{\text{trace}} \\
	   \lstick{$\ket{\bar{0}}_{\sfY}$} & \qw & \qwbundle[]{} & \gate[2]{Q_1} & \qwbundle[]{} & \gate[2]{Q_0} & \qwbundle[]{} & \qw\\
	   \lstick{$\ket{\bar{0}}_{\sfZ}$} & \qw & \qwbundle[]{} & & \qwbundle[]{} & & \qwbundle[]{} & \trash{\text{trace}}
    \end{quantikz}
    \caption{Quantum circuit $Q'_1$. }
    \label{fig:QJSP-in-QSZK-Qone}
    \end{minipage}
    \end{figure}	
    
    We then obtain the following: 
	\begin{equation}
	\label{eq:QJS-to-QED}
	\begin{aligned}
		\S_2(\rho'_0)-\S_2(\rho'_1) &= \S_2(\sfB',\sfY)_{\rho'_0} - \S_2(\sfB,\sfY)_{\rho'_1}\\
		&= \rbra*{\H(\sfB')+\S_2(\sfY|\sfB')_{\rho'_0}}-\rbra*{\H(\sfB)+\S_2(\sfY|\sfB)_{\rho'_1}}\\
		&= \S_2(\sfY)_{\rho'_0} - \S_2(\sfY|\sfB)_{\rho'_1} + \H(\sfB') - \H(\sfB)\\
		&= \S_2(\sfY)_{\rho'_0}-\S_2(\sfY|\sfB)_{\rho'_1} -\frac{\alpha(n)+\beta(n)}{2}\\
		&= \S_2\rbra*{ \frac{\rho_0+\rho_1}{2} } - \frac{\S_2(\rho_0)+\S_2(\rho_1)}{2} -\frac{\alpha(n)+\beta(n)}{2}\\
		&= \QJS_2(\rho_0,\rho_1) - \frac{\alpha(n)+\beta(n)}{2}, 
	\end{aligned}
	\end{equation}
	where the second line is due to the definition of quantum conditional entropy and both $\sfB$ and $\sfB'$ are classical registers, the third line owes to the fact that $\sfB'$ is an independent random bit, the fifth line follows from the joint entropy theorem, such as~\cite[Theorem 11.8(5)]{NC10}.
 	
	Plugging \Cref{eq:QJS-to-QED} into the promise of $\QJSP[\alpha,\beta]$,  we obtain the following and choose $g(n')=\frac{\ln{2}}{2}(\alpha(n)-\beta(n))$: 
    \begin{itemize}[topsep=0.33em, itemsep=0.33em, parsep=0.33em]
        \item If $\QJS_2(\rho_0,\rho_1) \geq \alpha(n)$, then $\S(\rho'_0)-\S(\rho'_1) \geq \frac{\ln{2}}{2}(\alpha(n)-\beta(n)) = g(n')$; 
		\item If $\QJS_2(\rho_0,\rho_1) \leq \beta(n)$, then $\S(\rho'_0)-\S(\rho'_1) \leq -\frac{\ln{2}}{2}(\alpha(n)-\beta(n)) = -g(n')$. 
    \end{itemize}

    By inspecting the description of quantum circuits $Q'_0$ and $Q'_1$, we know that the number of output qubit is $n' \coloneqq n+1$ and these circuits act on at most $m'(n') = m(n)+3$ qubits. Therefore, $\QJSP[\alpha,\beta]$ is Karp reducible to $\QEDP[g(n)]$ by mapping $(Q_0,Q_1)$ to $(Q'_0,Q'_1)$. 
\end{proof}

\subsection{\QSZK{} containments via the polarization lemma}
\label{subsec:polarization-lemma-QTD}

We introduce new polarization lemmas for the measured quantum triangular discrimination ($\measQTD$) and the quantum triangular discrimination (\QTD{}), as stated in \Cref{lemma:measQTD-polarization,lemma:QTD-polarization}, respectively. These techniques can also be used to establish \QSZK{} containments of both \measQTDP{} and \QTDP{}. A notable feature of this approach is that the polarization lemma for $\measQTD$ requires only the condition $\alpha > \beta$, in contrast to the parameter requirements for the trace distance and \QTD{}, which demand the stronger condition $\alpha^2 > \beta$. 

\begin{lemma}[A polarization lemma for $\measQTD$]
\label{lemma:measQTD-polarization}
Given quantum circuits $Q_0$ and $Q_1$ that prepare quantum states $\rho_0$ and $\rho_1$, respectively, there exists a deterministic procedure that takes as input $(Q_0,Q_1,\alpha,\beta,k)$, where $\alpha > \beta$, and outputs quantum circuits $\tilde{Q}_0$ and $\tilde{Q}_1$, which prepare quantum states $\tilde{\rho}_0$ and $\tilde{\rho}_1$, respectively. The resulting states satisfy\emph{:} 
\begin{align*}
    \measQTD(\rho_0,\rho_1) \geq \alpha  \quad &\Longrightarrow \quad  \measQTD(\trho_0,\trho_1) \geq 1-2^{-k},\\
    \measQTD(\rho_0,\rho_1) \leq \beta  \quad &\Longrightarrow \quad  \measQTD(\trho_0,\trho_1) \leq 2^{-k}. 
\end{align*}
Here, the states $\trho_0$ and $\trho_1$ are defined over $\widetilde{O}\Big( n k^{O\big(\frac{\beta \ln(2/\alpha)}{\alpha-\beta}\big)} \Big)$ qubits. Furthermore, when $k \leq O(1)$ or $\alpha - \beta \geq \Omega(1)$, the time complexity of the procedure is polynomial in the size of $Q_0$ and $Q_1$, $k$, and $\exp\big(\frac{\beta \log(1/\alpha)}{\alpha-\beta}\big)$. 
\end{lemma}

\begin{lemma}[A polarization lemma for \QTD{}]
\label{lemma:QTD-polarization}
Given quantum circuits $Q_0$ and $Q_1$ that prepare quantum states $\rho_0$ and $\rho_1$, respectively, there exists a deterministic procedure that takes as input $(Q_0,Q_1,\alpha,\beta,k)$, where $\alpha^2 > \beta$, and outputs quantum circuits $\tilde{Q}_0$ and $\tilde{Q}_1$, which prepare quantum states $\tilde{\rho}_0$ and $\tilde{\rho}_1$, respectively. The resulting states satisfy\emph{:}  
\begin{align*}
    \QTD(\rho_0,\rho_1) \geq \alpha  \quad &\Longrightarrow \quad  \QTD(\trho_0,\trho_1) \geq 1-2^{-k},\\
    \QTD(\rho_0,\rho_1) \leq \beta  \quad &\Longrightarrow \quad  \QTD(\trho_0,\trho_1) \leq 2^{-k}. 
\end{align*}
Here, the states $\trho_0$ and $\trho_1$ are defined over $\widetilde{O}\Big( n k^{O\big(\frac{\beta \ln(2/\alpha^2)}{\alpha^2-\beta}\big)} \Big)$ qubits. Furthermore, when $k \leq O(1)$ or $\alpha^2 - \beta \geq \Omega(1)$, the time complexity of the procedure is polynomial in the size of $Q_0$ and $Q_1$, $k$, and $\exp\big(\frac{\beta \log(1/\alpha^2)}{\alpha^2-\beta}\big)$.  
\end{lemma}

Analogous to the \QSZK{} containment of \QSDP{}, we can establish \QSZK{} containments of \measQTDP{} and \QTDP{} by leveraging their respective polarization lemmas: 
\begin{lemma}[\measQTDP{} and \QTDP{} are in \QSZK{}]
\label{lemma:measQTDP-in-QSZK}
Let $\alpha(n)$ and $\beta(n)$ be efficiently computable functions satisfying $0 \leq \beta < \alpha \leq 1$. Then, the following holds\emph{:}
\begin{enumerate}[label={\upshape(\roman*)}, topsep=0.33em, itemsep=0.33em, parsep=0.33em]
    \item For any $\alpha(n) - \beta(n) \geq 1/O(\log{n})$,  $\measQTDP[\alpha,\beta]$ is in \QSZK{}.
    \item For any $\alpha^2(n) - \beta(n) \geq 1/O(\log{n})$, $\QTDP[\alpha,\beta]$ is in \QSZK{}.
\end{enumerate}
\end{lemma}

\begin{proof}
For any $\measQTDP[\alpha,\beta]$ instance satisfying $\alpha(n)-\beta(n) \geq 1/O(\log{n})$, the polarization lemma for $\measQTD$ (\Cref{lemma:measQTD-polarization}) enables mapping it to a $\measQTDP[1-2^{-l(n)},2^{-l(n)}]$ instance, where $2^{-l(n)}$ is a negligible function. 
Similarly, for any $\QTDP[\alpha,\beta]$ instance with $\alpha^2(n) - \beta(n) \geq 1/O(\log{n})$, the polarization lemmas for \QTD{} (\Cref{lemma:QTD-polarization}) allows mapping it to a $\measQTDP[1-2^{-l(n)},2^{-l(n)}]$ instance. 
Using the inequalities in \Cref{thm:QTD-vs-td}, we establish reductions from \measQTDP{} and \QTDP{} to \QSDP{}:
\begin{itemize}[topsep=0.33em, itemsep=0.33em, parsep=0.33em, leftmargin=2em]
    \item For \textit{yes} instances, it holds that 
    \[\td(\rho_0,\rho_1) \geq \measQTD(\rho_0, \rho_1) \geq 1-2^{-l} \quad \text{and} \quad \td(\rho_0,\rho_1) \geq \QTD(\rho_0, \rho_1) \geq 1-2^{-l}.\]
    \item For \textit{no} instances, the inequality $\td(\rho_0,\rho_1) \leq 2^{-l/2}$ is guaranteed by 
    \[ \td^2(\rho_0,\rho_1) \leq \measQTD(\rho_0,\rho_1) \leq 2^{-l} \quad \text{and} \quad \td^2(\rho_0,\rho_1) \leq \QTD(\rho_0,\rho_1) \leq 2^{-l}. \]    
\end{itemize}

Finally, by following \cite[Theorem 10]{Wat02}, specifically the protocol in \cite[Figure 2]{Wat02}, we conclude that $\measQTDP[1-2^{-l(n)},2^{-l(n)}]$ and $\QTDP[1-2^{-l(n)},2^{-l(n)}]$ are indeed contained in \QSZK{}. 
\end{proof}

\subsubsection{Polarization lemmas for \texorpdfstring{$\measQTD$}{} and \texorpdfstring{\QTD{}}{}}
We now establish the polarization lemmas for $\measQTD$ (\Cref{lemma:measQTD-polarization}) and \QTD{} (\Cref{lemma:QTD-polarization}). The proofs rely on two independent one-sided error reduction techniques: one for \textit{no} instances and another for \textit{yes} instances, which are applied separately and in alternation.

\paragraph{No-instance error reduction for \measQTDP{} and \QTDP{}.} We begin with \textit{no}-instance error reduction, which is referred to as the XOR lemma in the polarization lemma for \SDP{}. 
It is worth noting that the corresponding statements for both $\measQTD$ and \QTD{} involve the same type of identity.

\begin{lemma}[No-instance error reduction for \measQTDP{} and \QTDP{}]
\label{lemma:measQTD-no-error-reduction}
Given quantum circuits $Q_0$ and $Q_1$ that prepare the quantum states $\rho_0$ and $\rho_1$, respectively, there exists a deterministic procedure that, on input $(Q_0,Q_1,l)$, produces new quantum circuits $\tilde{Q}_0$ and $\tilde{Q}_1$ preparing the states $\trho_0$ and $\trho_1$, respectively. 
These states are defined as $\trho_b=2^{-l+1}\sum_{b_1\oplus\cdots\oplus b_l=b} \rho_{b_1}\otimes \cdots \otimes \rho_{b_l} \text{ for } b\in\binset$, and satisfy the following identities\emph{:}
\[\measQTD(\trho_0,\trho_1) = \measQTD(\rho_0,\rho_1)^l \quad \text{and} \quad \QTD(\trho_0,\trho_1) = \QTD(\rho_0,\rho_1)^l.  \]
\end{lemma}

\begin{proof}
It suffices to prove that for quantum states $\rho_0$, $\rho_1$, $\rho'_0$, and $\rho'_1$, defining 
\[\trho_0 \coloneqq \frac{1}{2}(\rho'_0\otimes \rho_0 + \rho'_1\otimes \rho_1) \text{ and } \trho_1 \coloneqq \frac{1}{2}(\rho'_0\otimes \rho_1 + \rho'_1\otimes \rho_0),\]
the following identities hold:   
\begin{align}
    \measQTD(\trho_0,\trho_1) &= \measQTD(\rho'_0,\rho'_1) \cdot \measQTD(\rho_0,\rho_1), \label{eq:measQTD-no-error-reduction}\\
    \QTD(\trho_0,\trho_1) &= \QTD(\rho'_0,\rho'_1) \cdot \QTD(\rho_0,\rho_1). \label{eq:QTD-no-error-reduction}
\end{align}
Hence, we can conclude the proof by inductively applying \Cref{eq:measQTD-no-error-reduction} to $\measQTD\big(\trho^{(l)}_0,\trho^{(l)}_1\big)$, and \Cref{eq:QTD-no-error-reduction} to $\QTD\big(\trho^{(l)}_0,\trho^{(l)}_1\big)$.

It remains to demonstrate the identities in \Cref{eq:measQTD-no-error-reduction,eq:QTD-no-error-reduction}. For \Cref{eq:measQTD-no-error-reduction}, mirroring the approach of~\cite[Proposition 4.12]{BDRV19}, we obtain:
\begin{align*}
\measQTD(\trho_0,\trho_1) 
=& \sup_{{\rm POVM}~\calE} \TD\left(\tilde{p}_0^{(\calE)}, \tilde{p}_1^{(\calE)}\right)\\
=& \sup_{{\rm POVM}~\calE} \TD\left(p_0^{(\calE)}, p_1^{(\calE)}\right) \cdot \TD\left({p'_0}^{(\calE)}, {p'_1}^{(\calE)}\right)\\
=& \sup_{{\rm POVM}~\calE_1} \TD\left(p_0^{(\calE_1)}, p_1^{(\calE_1)}\right) \cdot \sup_{{\rm POVM}~\calE_2} \TD\left({p'_0}^{(\calE_2)}, {p'_1}^{(\calE_2)}\right)\\
=& \measQTD(\rho_0,\rho_1) \cdot \measQTD(\rho'_0,\rho'_1). 
\end{align*}

For \Cref{eq:QTD-no-error-reduction}, the identity follows from the identities: 
\[ \trho_0-\trho_1 = \frac{1}{2} \rbra*{\rho'_0-\rho'_1}\otimes\rbra*{\rho_0-\rho_1} \quad \text{and} \quad \trho_0+\trho_1 = \frac{1}{2} \rbra*{\rho'_0+\rho'_1}\otimes\rbra*{\rho_0+\rho_1}. \qedhere\]
\end{proof}

\paragraph{Yes-instance error reduction for \measQTDP{} and \QTDP{}.}
We then proceed with \textit{yes}-instance error reduction, known as the direct product lemma in polarization lemma for \SDP{}. 
Notably, the $\measQTD$ case (\Cref{lemma:measQTD-yes-error-reduction}) achieves a lower bound with a \textit{quadratic} improvement compared to both the trace distance case~\cite[Lemma 9]{Wat02} and the \QTD{} case (\Cref{lemma:QTD-yes-error-reduction}), whereas the upper bound is slightly worse than the trace distance case.\footnote{This difference arises from the fact that the trace distance and statistical distance are metrics, while the triangular discrimination and its quantum analogs ($\measQTD$ and \QTD{}) are (conjectured to be) the \textit{squares} of a metric.}

\begin{lemma}[Yes-instance error reduction for \measQTDP{}]
\label{lemma:measQTD-yes-error-reduction}
Given quantum circuits $Q_0$ and $Q_1$ that prepare the quantum states $\rho_0$ and $\rho_1$, respectively, there exists a deterministic procedure that, on input $(Q_0,Q_1,l)$, produces new quantum circuits $\tilde{Q}_0$ and $\tilde{Q}_1$ preparing the states $\trho_0$ and $\trho_1$. These states are defined as $\trho_b \coloneqq \rho_b^{\otimes l}$ for $b\in\binset$, and satisfy the inequalities\emph{:} 
\[1-\exp\rbra*{-\frac{l}{2} \cdot \measQTD(\rho_0,\rho_1)} \leq \measQTD(\trho_0, \trho_1)  \leq 2l \cdot \measQTD(\rho_0,\rho_1).\] 
\end{lemma}

\begin{proof}
The proof follows the approach of~\cite[Lemma 4.10]{BDRV19}, utilizing a key property of the Bures distance on tensor-product states $\rho_0^{\otimes l}$ and $\rho_1^{\otimes l}$: 
\begin{equation}
    \label{eq:Bures-tensor-product}
    \frac{1}{2}\Bsquare\big(\rho_0^{\otimes l},\rho_1^{\otimes l}\big) = 1-\F\rbra*{\rho_0^{\otimes l},\rho_1^{\otimes l}} = 1-\F\rbra*{\rho_0,\rho_1}^l =1-\left(1-\frac{1}{2}\Bsquare(\rho_0,\rho_1)\right)^l.
\end{equation}

By utilizing \Cref{lemma:B-leq-measQTD-leq-2B}, we obtain the following upper bound: 
\begin{align*}
\measQTD(\rho_0^{\otimes l},\rho_1^{\otimes l}) 
&\leq \Bsquare(\rho_0^{\otimes l},\rho_1^{\otimes l})\\
&=2\left(1-\left(1-\frac{1}{2}\Bsquare(\rho_0,\rho_1)\right)^l\right)\\
&\leq l\Bsquare(\rho_0,\rho_1)\\ 
&\leq 2l\measQTD(\rho_0,\rho_1),
\end{align*}
where the third line is because $(1-x)^k \geq 1-kx$ for any $x$ and integer $k$. Likewise, we can also deduce the following lower bound: 
\begin{align*}
\measQTD(\rho_0^{\otimes l},\rho_1^{\otimes l}) 
&\geq \frac{1}{2}\Bsquare(\rho_0^{\otimes l},\rho_1^{\otimes l}) \\
&= \left(1-\left(1-\frac{1}{2}\Bsquare(\rho_0,\rho_1)\right)^l\right)\\ 
&\geq \left(1-\left(1-\frac{1}{2}\measQTD(\rho_0,\rho_1)\right)^l\right)\\
&\geq 1-\exp\left(-\frac{l}{2}\measQTD(\rho_0,\rho_1)\right),
\end{align*}
where the last equality owes to $1-x\leq e^{-x}$ for any $x$. These bounds complete the proof. 
\end{proof}

Interestingly, the lower bound in \Cref{lemma:QTD-yes-error-reduction} matches that of the trace distance case, even though the proof techniques differ. The trace distance case relies on the triangle inequality, which is only \textit{conjectured} to hold for $\sqrt{\QTD}$. In contrast, our proof circumvents this barrier by leveraging the inequalities between \QTD{} and the Bures distance. 

\begin{lemma}[Yes-instance error reduction for \QTDP{}]
\label{lemma:QTD-yes-error-reduction}
Given quantum circuits $Q_0$ and $Q_1$ that prepare the quantum states $\rho_0$ and $\rho_1$, respectively, there exists a deterministic procedure that, on input $(Q_0,Q_1,l)$, produces new quantum circuits $\tilde{Q}_0$ and $\tilde{Q}_1$ preparing the states $\trho_0$ and $\trho_1$. These states are defined as $\trho_b \coloneqq \rho_b^{\otimes l}$ for $b\in\binset$, and satisfy the inequalities\emph{:} 
\[1-\exp\rbra*{-\frac{l}{2} \cdot \QTD(\rho_0,\rho_1)^2} \leq \QTD(\trho_0, \trho_1)  \leq \sqrt{2l} \cdot \sqrt{\QTD(\rho_0,\rho_1)}.\] 
\end{lemma}

\begin{proof}
    Our proof strategy closely follows the approach used in \Cref{lemma:measQTD-yes-error-reduction}. For the upper bound, we use the inequalities from \Cref{prop:QTD-leq-B} and \Cref{eq:Bures-tensor-product}, which give
    \[ \QTD\rbra*{\rho_0^{\otimes l}, \rho_1^{\otimes l}} 
    \leq \B\rbra*{\rho_0^{\otimes l}, \rho_1^{\otimes l}}
    \leq \sqrt{l} \cdot \B(\rho_0,\rho_1)
    \leq \sqrt{2l} \cdot \sqrt{\QTD(\rho_0,\rho_1)}.\]
    
    For the lower bound, we again apply \Cref{prop:QTD-leq-B} and \Cref{eq:Bures-tensor-product}, obtaining
    \begin{align*}
        \QTD\rbra*{\rho_0^{\otimes l}, \rho_1^{\otimes l}} 
        \geq \frac{1}{2} \Bsquare\rbra*{\rho_0^{\otimes l}, \rho_1^{\otimes l}} 
        &= 1 - \rbra*{1 - \frac{1}{2} \Bsquare(\rho_0,\rho_1)}^l \\
        &\geq 1 - \rbra*{1 - \frac{1}{2} \QTD(\rho_0,\rho_1)^2}^l \\
        &\geq 1-\exp\rbra*{-\frac{l}{2} \cdot \QTD(\rho_0,\rho_1)^2}. \qedhere
    \end{align*}
\end{proof}

\paragraph{Putting everything together.} We can now establish \Cref{lemma:measQTD-polarization,lemma:QTD-polarization} by selecting appropriate parameters based on the polarization lemma for the triangular discrimination, as established in~\cite[Lemma 4.9]{BDRV19}. 
Specifically, we first apply \textit{no}-instance error reduction (\Cref{lemma:measQTD-no-error-reduction}), then use \textit{yes}-instance error reduction (\Cref{lemma:measQTD-yes-error-reduction} or \Cref{lemma:QTD-yes-error-reduction}) to ensure that the soundness parameter is at most $1/2$, and finally apply \textit{no}-instance error reduction (\Cref{lemma:measQTD-no-error-reduction}) again. The time complexity analysis aligns with~\cite[Lemma 38]{CCKV08}. 

\begin{proof}[Proof of \Cref{lemma:measQTD-polarization}]
    Let $\lambda \coloneqq \min\cbra*{\alpha/\beta,2} \in (1,2]$, and choose $l \coloneqq \lceil \log_{\lambda} 8k \rceil$. 
    Applying the \textit{no}-instance error reduction for \measQTDP{} (\Cref{lemma:measQTD-no-error-reduction}) to the input $(Q_0,Q_1,l)$, where the quantum circuits $Q_0$ and $Q_1$ prepare the states $\rho_0$ and $\rho_1$, respectively, produces new quantum circuits $(Q'_0,Q'_1)$ with corresponding states $(\rho'_0,\rho'_1)$ such that:
    \begin{align*}
    \measQTD(\rho_0,\rho_1) \geq \alpha \quad &\Longrightarrow \quad \measQTD(\rho'_0,\rho'_1)\geq \alpha^l; \\
    \measQTD(\rho_0,\rho_1) \leq \beta \quad &\Longrightarrow \quad \measQTD(\rho'_0,\rho'_1)\leq \beta^l.
    \end{align*}

    Let $m \coloneqq \lambda^l/(4\alpha^l) \leq 1/(4\beta^l)$, and define the states $\rho''_0 \coloneqq (\rho'_0)^{\otimes m}$ and $\rho''_1 \coloneqq (\rho'_1)^{\otimes m}$, along with the corresponding circuits $Q''_0$ and $Q''_1$. Applying the \textit{yes}-instance error reduction for \measQTDP{} (\Cref{lemma:measQTD-yes-error-reduction}) to the input $(Q'_0, Q'_1,m)$ yields that: 
    \begin{align*}
    \measQTD(\rho_0,\rho_1) \geq \alpha \quad &\Longrightarrow \quad \measQTD(\rho''_0,\rho''_1)\geq 1-\exp\rbra*{-\alpha^l m/2} \geq 1-e^{-k}; \\
    \measQTD(\rho_0,\rho_1) \leq \beta \quad &\Longrightarrow \quad \measQTD(\rho''_0,\rho''_1)\leq 2m\beta^l \leq 1/2. 
    \end{align*}

    Finally, applying the \textit{no}-instance error reduction for \measQTDP{} (\Cref{lemma:measQTD-no-error-reduction}) again to the input $(Q''_0, Q''_1, k)$ produces new quantum circuits $(\tilde{Q}_0,\tilde{Q}_1)$ with the corresponding quantum states $(\trho_0,\trho_1)$, satisfying:
    \begin{align*}
    \measQTD(\rho_0,\rho_1) \geq \alpha \quad &\Longrightarrow \quad \measQTD(\trho_0,\trho_1)\geq \rbra*{1-e^{-k}}^k \geq 1-ke^{-k} \geq 1-2^{-k}; \\
    \measQTD(\rho_0,\rho_1) \leq \beta \quad &\Longrightarrow \quad \measQTD(\trho_0,\trho_1)\leq 2^{-k}.
    \end{align*}
    
    The last step holds for sufficiently large $k$, which we can be determined by selecting an appropriate value at the beginning of our construction.  

    The time complexity analysis follows a similar approach to~\cite[Lemma 38]{CCKV08}. Specifically, noting that $\lambda \in (1,2]$, we have $\ln(\lambda) = \ln(1+(\lambda-1))\geq (\lambda-1)/2 \geq \Omega\big(\frac{\alpha-\beta}{\beta}\big)$, where the first inequality is due to $\ln(1+x)\geq x/2$ for all $x\in[0,1]$. Then, we obtain $l=O\big(\frac{\ln k}{\ln \lambda}\big)=O\big(\frac{\beta \ln{k}}{\alpha-\beta}\big)$ and further conclude that $m \leq (2/\alpha)^l/4 = \exp\big( O\big( \frac{\beta \ln k}{\alpha-\beta} \cdot \ln (2/\alpha) \big) \big)$. 
\end{proof}

\begin{proof}[Proof of \Cref{lemma:QTD-polarization}]
    Our proof strategy closely mirrors the approach used in \Cref{lemma:measQTD-polarization}, but with different parameters $\lambda,l,m$, and some intermediate steps are omitted for brevity. 
    
    We set $\lambda \coloneqq \min\cbra*{\alpha^2/\beta,2} \in (1,2]$, and choose $l \coloneqq \ceil*{\log_{\lambda}(16k)}$. By applying the \textit{no}-instance error reduction for \QTDP{} (\Cref{lemma:measQTD-no-error-reduction}) to the input $(Q_0,Q_1,l)$, we obtain the circuits $(Q'_0,Q'_1)$ and the corresponding states $(\rho'_0,\rho'_1)$, satisfying: 
    \begin{align*}
        \QTD(\rho_0,\rho_1) \geq \alpha \quad &\Longrightarrow \quad \QTD(\rho'_0,\rho'_1)\geq \alpha^l; \\
        \QTD(\rho_0,\rho_1) \leq \beta \quad &\Longrightarrow \quad \QTD(\rho'_0,\rho'_1)\leq \beta^l.
    \end{align*}

    Next, let $m \coloneqq \lambda^l/(8 \alpha^{2l}) \leq 1/(8\beta^l)$. Applying the \textit{yes}-instance error reduction for \QTDP{} (\Cref{lemma:QTD-yes-error-reduction}) to the input $(Q'_0, Q'_1,m)$, where the resulting circuits and states are denoted by $(Q''_0,Q''_1)$ and $(\rho''_0,\rho''_1)$, respectively, yields the following: 
    \begin{align*}
        \QTD(\rho_0,\rho_1) \geq \alpha \quad &\Longrightarrow \quad \QTD(\rho''_0,\rho''_1)\geq 1-\exp\rbra*{-\alpha^{2l} m/2} \geq 1-e^{-k}; \\
        \QTD(\rho_0,\rho_1) \leq \beta \quad &\Longrightarrow \quad \QTD(\rho''_0,\rho''_1)\leq \sqrt{2m}\beta^{l/2} \leq 1/2. 
    \end{align*}

    Lastly, applying the \textit{no}-instance error reduction for \QTDP{} (\Cref{lemma:measQTD-no-error-reduction}) again to the input $(Q''_0, Q''_1, k)$ results in the circuits $(\tilde{Q}_0,\tilde{Q}_1)$ and the corresponding quantum states $(\trho_0,\trho_1)$, where the following holds: 
    \begin{align*}
        \QTD(\rho_0,\rho_1) \geq \alpha \quad &\Longrightarrow \quad \QTD(\trho_0,\trho_1)\geq \rbra*{1-e^{-k}}^k \geq 1-ke^{-k} \geq 1-2^{-k}; \\
        \QTD(\rho_0,\rho_1) \leq \beta \quad &\Longrightarrow \quad \QTD(\trho_0,\trho_1)\leq 2^{-k}.
    \end{align*}

    The time complexity analysis follows similarly to the proof of \Cref{lemma:measQTD-polarization}. 
    Since $\lambda \in (1,2]$, we obtain $\ln{\lambda} \geq \Omega\rbra[\big]{\frac{\alpha^2-\beta}{\beta}}$, and thus $l = O\rbra[\big]{\frac{\ln{k}}{\ln{\lambda}}} = O\rbra[\big]{\frac{\beta\ln{k}}{\alpha^2-\beta}}$. Consequently, we conclude that 
    $m \leq \rbra*{2/\alpha^2}^l/8 \leq \exp\rbra[\big]{O\rbra[\big]{\frac{\beta \ln{k}}{\alpha^2-\beta} \cdot \ln(2/\alpha^2)}}$.
\end{proof}

\subsection{\QSZK{}-hardness for \QJSP{}, \QEDP{}, \measQTDP{}, and \QTDP{}}
\label{subsec:QSZK-hardness}

\paragraph{\QJSP{} is \QSZK{}-hard.} We begin by establishing the \QSZK{}-hardness of the \textsc{Quantum Jensen-Shannon Divergence Problem} (\QJSP{}):

\begin{lemma}[\QJSP{} is \QSZK{}-hard]
    \label{lemma:QJSP-is-QSZK-hard}
    Let $\alpha(n)$ and $\beta(n)$ be efficiently computable functions, there exists a constant $\epsilon\in(0,1/2)$ such that $\QJSP[\alpha,\beta]$ is \QSZK{}-hard when $\alpha(n) \leq 1-2^{-n^{1/2-\epsilon}}$ and $\beta(n)\geq 2^{-n^{1/2-\epsilon}}$ for sufficiently large $n$. 
\end{lemma}

Following the approach for proving \SZK{}-hardness in \JSP{}~\cite[Lemma 4.3]{BDRV19}, we prove \Cref{lemma:QJSP-is-QSZK-hard} by utilizing inequalities between the trace distance and $\QJS_2$ (combining \Cref{lemma:traceDist-leq-QJS,lemma:QJS-leq-traceDist}), which mirror the inequalities between the statistical distance and the Jensen-Shannon divergence~\cite{FvdG99,Top00}. 

\begin{proof}[Proof of \Cref{lemma:QJSP-is-QSZK-hard}]
    Using \Cref{thm:QSD-is-QSZKhard}, it suffices to reduce $\QSDP\big[1-2^{-n^{1/2-\epsilon/2}},2^{-n^{1/2-\epsilon/2}}\big]$ to $\QJSP[\alpha,\beta]$, where $\alpha$ and $\beta$ will be specified later. 
    Consider quantum circuits $Q_0$ and $Q_1$ acting on $m(n)$ qubits, which is a \QSDP{} instance. We can obtain $\rho_i$ for $i\in\binset$ by performing $Q_i$ on $\ket{0^m}$ and tracing out the non-output qubits. This yields the following: 
    \begin{itemize}[topsep=0.33em, itemsep=0.33em, parsep=0.33em]
        \item If $\td(\rho_0,\rho_1) \geq 1-2^{-n^{1/2-\epsilon/2}}$, then \Cref{lemma:traceDist-leq-QJS} indicates that
        \begin{align*}
            \QJS_2(\rho_0,\rho_1) \geq 1-\binH\left(\frac{1-\td(\rho_0,\rho_1)}{2}\right)
            &\geq 1-\binH\left(2^{-n^{1/2-\epsilon/2}-1}\right)\\
            &\geq 1-2\cdot 2^{-(n^{1/2-\epsilon/2}+1)/2}\\
            &\geq \alpha(n),
        \end{align*}
        where the third inequality owes to $\binH(x)\leq 2\sqrt{x}$ for all $x\in[0,1]$. Then we choose a constant $n(\epsilon)$ such that the last inequality holds. Specifically, there exists a constant $n(\epsilon)$ such that $1-2\cdot 2^{-(n^{1/2-\epsilon/2}+1)/2} \geq 1-2^{-n^{1/2-\epsilon}}$ for all $n\geq n(\epsilon)$. 
        
        \item If $\td(\rho_0,\rho_1) \leq 2^{-n^{1/2-\epsilon/2}}$, then according to \Cref{lemma:QJS-leq-traceDist}, we have
        \[\QJS_2(\rho_0,\rho_1) \leq \td(\rho_0,\rho_1) \leq 2^{-n^{1/2-\epsilon/2}} \leq \beta(n).\]
        Here, the last inequality holds for any $n\geq n(\epsilon)$ since $\beta(n) \geq 2^{-n^{1/2-\epsilon/2}}$.
    \end{itemize}

    Therefore, by utilizing the same quantum circuits $Q_0$ and $Q_1$ and their corresponding states $\rho_0$ and $\rho_1$, we establish a Karp reduction from $\QSDP\big[1-2^{-n^{1/2-\epsilon/2}},2^{-n^{1/2-\epsilon/2}}\big]$ to $\QJSP[\alpha,\beta]$ for $n \geq n(\epsilon)$. 
\end{proof}

\paragraph{A simple \QSZK{}-hardness proof for \QEDP{}.}
Furthermore, we can establish a new and simple reduction from \QSDP{} to \QEDP{} via \QJSP{} by combining \Cref{lemma:QJSP-is-in-QSZK,lemma:QJSP-is-QSZK-hard}. This reduction leads to a simple \QSZK{}-hardness proof for \QEDP{}, as stated in \Cref{corr:simple-QSZKhardness-QEDP}. 
Now we present the detailed proof: 

\begin{proof}[Proof of \Cref{corr:simple-QSZKhardness-QEDP}]
Using \Cref{lemma:QJSP-is-QSZK-hard}, we obtain that $\QJSP[\alpha,\beta]$ is \QSZK{}-hard when $\alpha(n) \leq 1-2^{-n^{1/2-\epsilon}}$ and $\beta(n) \geq 2^{-n^{1/2-\epsilon}}$ for some $\epsilon\in(0,1/2)$ and $n \geq n(\epsilon)$. 
The hard instances for \QSDP{} (simultaneously hard for \QJSP{}), as specified in \Cref{lemma:QJSP-is-QSZK-hard}, consist of quantum circuits $Q_0$ and $Q_1$, acting on $m(n)$ qubits, that prepare a purification of $n$-qubit quantum states $\rho_0$ and $\rho_1$, respectively. 

Subsequently, by \Cref{lemma:QJSP-is-in-QSZK}, we construct quantum circuits $Q'_0$ and $Q'_1$ acting on $m'(n')=m(n)+3$ qubits, where $n' \coloneqq n+1$, preparing a purification of $n'$-qubit states 
$\rho'_0 = (p\ket{0}\bra{0}+(1-p)\ket{1}\bra{1})\otimes (\frac{1}{2}\rho_0 + \frac{1}{2}\rho_1)$ satisfying $\H_2(p) = 1-\frac{\ln{2}}{2}(\alpha+\beta)$ and $\rho'_1 = \frac{1}{2} \ket{0}\bra{0}\otimes \rho_0 + \frac{1}{2}\ket{1}\bra{1}\otimes \rho_1$, where $r'(n')=r(n)+1$. According to \Cref{lemma:QJSP-is-in-QSZK}, $\QEDP[g]$ is \QSZK{}-hard as long as
\[g(n') = \frac{\ln{2}}{2}(\alpha(n'-3) - \beta(n'-3)) 
\leq \frac{\ln{2}}{2} \big( 1-2^{-(n'-3)^{1/2-\epsilon}+1} \big).\]
\QSDP{} is thus Karp reducible to \QEDP{} by mapping $(Q_0,Q_1)$ to $(Q'_0,Q'_1)$. To finish the proof, we redefine $n \coloneqq n'$, replacing $n'$ with $n$ in the \QSZK{}-hardness condition for \QEDP{}. 
\end{proof}

\paragraph{\measQTDP{} and \QTDP{} are \QSZK{}-hard}
Next, we prove the \QSZK{}-hardness of both the \textsc{Measured Quantum Triangular Discrimination Problem} (\measQTDP{}) and the \textsc{Quantum Triangular Discrimination Problem} (\QTDP{}): 

\begin{lemma}[\measQTDP{} and \QTDP{} are \QSZK{}-hard]
    \label{lemma:measQTDP-is-QSZK-hard}
    Let $\alpha(n)$ and $\beta(n)$ be efficiently computable functions, there exists a constant $\epsilon\in(0,1/2)$ such that
    \[\measQTDP[\alpha,\beta] \text{ and } \QTDP[\alpha,\beta] \text{ are } \QSZK{}\text{-hard},\] 
    when $\alpha(n)\leq 1-2^{-n^{1/2-\epsilon}}$ and $\beta(n) \geq 2^{-n^{1/2-\epsilon}}$ for sufficiently large $n$. 
\end{lemma}

The proof parallels the approach to show the \SZK{}-hardness for \TDP{}~\cite[Lemma 4.4]{BDRV19}. We employ the inequalities between the trace distance and $\measQTD$ presented in \Cref{thm:QTD-vs-td}, analogous to the inequalities between the counterpart classical distances in~\cite{Top00}. 

\begin{proof}[Proof of \Cref{lemma:measQTDP-is-QSZK-hard}]
    Since the inequalities between the trace distance and \QTD{} coincides with those of $\measQTD$, we focus on proving that \measQTDP{} is \QSZK{}-hard in the desired regime. The proof can then be straightforwardly extended to the \QTDP{} case.
    
    By \Cref{thm:QSD-is-QSZKhard}, it suffices to reduce $\QSDP\big[1-2^{-n^{1/2-\epsilon/2}},2^{-n^{1/2-\epsilon/2}}\big]$ to $\measQTDP[\alpha,\beta]$, where $\alpha$ and $\beta$ will be specified later. 
    Consider quantum circuits $Q_0$ and $Q_1$ acting on $m(n)$ qubits, which is a \QSDP{} instance. 
    We can obtain $n$-qubit quantum states $\rho_i$ for $i\in\binset$ by performing $Q_i$ on $\ket{0^m}$ and tracing out the non-output qubits. This yields the following: 
    \begin{itemize}[topsep=0.33em, itemsep=0.33em, parsep=0.33em]
        \item If $\td(\rho_0,\rho_1) \geq 1-2^{-n^{1/2-\epsilon/2}}$, then \Cref{lemma:QSDsquare-leq-QTD} indicates that
        \[\measQTD(\rho_0,\rho_1) \geq \td(\rho_0,\rho_1)^2 \geq \left(1-2^{-n^{1/2-\epsilon/2}}\right)^2 
        \geq 1-2^{-n^{1/2-\epsilon/2}+1} \geq \alpha(n).\]
        We can choose a constant $n(\epsilon)$ such that $1-2^{-n^{1/2-\epsilon/2}+1} \geq 1-2^{-n^{1/2-\epsilon}}$ for all $n\geq n(\epsilon)$.

        \item If $\td(\rho_0,\rho_1) \leq 2^{-n^{1/2-\epsilon/2}}$, then  according to \Cref{lemma:QTD-leq-QSD} and \Cref{prop:measQTD-leq-QTD}, we have
        \[\measQTD(\rho_0,\rho_1) \leq \td(\rho_0,\rho_1) \leq 2^{-n^{1/2-\epsilon/2}} \leq \beta(n).\] 
        Here, the last inequality holds for any $n\geq n(\epsilon)$ because $\beta(n) \geq 2^{-n^{1/2-\epsilon/2}}$.
    \end{itemize}

    Therefore, by employing the same quantum circuits $Q_0$ and $Q_1$ and their corresponding states $\rho_0$ and $\rho_1$, we establish a Karp reduction from $\QSDP\big[1-2^{-n^{1/2-\epsilon/2}},2^{-n^{1/2-\epsilon/2}}\big]$ to $\measQTDP[\alpha,\beta]$ for $n \geq n(\epsilon)$. 
\end{proof}

%%%%%%%%%%%%%%%%%%%%%%%%%%%%%%%%%%%%%%%%%%%%%%%%%%%%%

\section{Easy regimes for the class \QSZK{}}
\label{sec:easy-regimes-QSZK}

We begin with the main results in this section: 

\begin{theorem}[Easy regimes for \QSZK{}]
\label{thm:easy-regimes-for-QSZK}
For any efficiently computable functions $\alpha$ and $\beta$, we have the following easy regimes for \QSZK{} in terms of \coQSDP{}: 
\begin{enumerate}[label={\upshape(\roman*)}, topsep=0.33em, itemsep=0.33em, parsep=0.33em]
    \item \label{thmitem:QSZK-easy-regimes-PP} $\coQSDP[\alpha,\beta]$ is in \PP{} when $1-2^{-n/2-1} \leq \alpha(n) \leq 1$ and $0 \leq \beta(n) \leq 2^{-n/2-1}$. 
    \item \label{thmitem:QSZK-easy-regimes-NQP} $\coQSDP[1,0]$ is in \NQP{}. 
\end{enumerate}
\end{theorem}

\Cref{thm:easy-regimes-for-QSZK} aligns with classical counterparts in terms of \SZK{}. 
In particular, \Cref{thm:easy-regimes-for-QSZK}\ref{thmitem:QSZK-easy-regimes-PP} is a quantum analog of~\cite[Theorem 7.1]{BCHTV19} which states that $\overline{\SDP}$ with some inverse-exponential errors is in \PP{}. Meanwhile, \Cref{thm:easy-regimes-for-QSZK}\ref{thmitem:QSZK-easy-regimes-NQP} parallels a folklore result that  $\overline{\SDP}$ without error is in \NP{}, as \NQP{} can be regarded as a quantum analog of \NP{}.

Furthermore, \Cref{thm:easy-regimes-for-QSZK}\ref{thmitem:QSZK-easy-regimes-PP} suggests that achieving a \textit{dimension-preserving} polarization for the \textsc{Quantum State Distinguishability Problem} (\QSDP{}) demands non-black-box techniques due to the existing oracle separation~\cite{BCHTV19}. This is because the existence of such a polarization would imply, by \Cref{thm:easy-regimes-for-QSZK}\ref{thmitem:QSZK-easy-regimes-PP}, that $\QSZK \subseteq \PP$. 

\subsection{\coQSDP{} without error is in \NQP{}}

As a prelude to \Cref{thm:easy-regimes-for-QSZK}\ref{thmitem:QSZK-easy-regimes-PP}, we will first establish \Cref{thm:easy-regimes-for-QSZK}\ref{thmitem:QSZK-easy-regimes-NQP}. Specifically, by making a crucial observation involving $\td(\rho_0,\rho_1)$ and $\Tr(\rho_0\rho_1)$, we can devise a unitary quantum algorithm $\calA$ based on the SWAP test (\Cref{lemma:swap-test}), which was  originally proposed for pure states in~\cite{BCWdW01} and later shown to be applicable to mixed states~\cite{KMY09}: 

\begin{lemma}[SWAP test for mixed states, adapted from~{\cite[Proposition 9]{KMY09}}]
\label{lemma:swap-test}
    Let $\rho_0$ and $\rho_1$ be two $n$-qubit quantum states, which may be mixed. There exists a $(2n+1)$-qubit quantum circuit that outputs $0$ with probability $\rbra*{1+\Tr(\rho_0\rho_1)}/2$, using a single copy of each $\rho_0$ and $\rho_1$ and $O(n)$ one- and two-qubit elementary quantum gates. 
\end{lemma}

The acceptance probability of $\calA$ is at least slightly higher than $1/2$ for \textit{yes} instances, while exactly $1/2$ for \textit{no} instances. 
We then apply exact amplitude amplification (\Cref{lemma:exact-amplitude-amplification}) on $\calA$ to construct another algorithm $\calA'$ that achieves one-sided error.

\begin{lemma} [Exact amplitude amplification, adapted from {\cite[Equation 8]{BHMT02}}]
\label{lemma:exact-amplitude-amplification}
    Suppose $U$ is a unitary operator such that $U\ket{\bar 0} = \sin(\theta) \ket{\psi_0} + \cos(\theta) \ket{\psi_1}$, where $\ket{\psi_0}$ and $\ket{\psi_1}$ are normalized pure states and $\innerprod{\psi_0}{\psi_1} = 0$.
    Let $G  \coloneqq  - U (I - 2\ket{\bar 0}\bra{\bar 0}) U^\dag (I - 2\ket{\psi_0}\bra{\psi_0})$ be the Grover operator. Then, for every integer $j \geq 0$, we have
    \[G^j U \ket{\bar 0} = \sin((2j+1)\theta) \ket{\psi_0} + \cos((2j+1)\theta) \ket{\psi_1}.\]
    Specifically, with a single iteration of $G$, we get $G U \ket{\bar 0} = \sin(3\theta) \ket{\psi_0} + \cos(3\theta) \ket{\psi_1}$.
\end{lemma}

\begin{proof}[Proof of \Cref{thm:easy-regimes-for-QSZK}\ref{thmitem:QSZK-easy-regimes-NQP}]
For any states $\rho_0$ and $\rho_1$, we can observe the following: 
\begin{itemize}[topsep=0.33em, itemsep=0.33em, parsep=0.33em]
    \item For \textit{yes} instances where $\td(\rho_0,\rho_1)=0$, we have $\rho_0=\rho_1$ due to the trace distance being a metric. This equality leads to $\Tr(\rho_0\rho_1)\geq 2^{-n}$, with equality achieved when both $\rho_0$ and $\rho_1$ correspond to the maximally mixed state $2^{-n} I_n$, where $I_n$ denotes the identity matrix on $n$ qubits.
    \item For \textit{no} instances where $\td(\rho_0,\rho_1)=1$, we know that $\rho_0$ and $\rho_1$ have orthogonal supports because of the triangle inequality, leading to $\Tr(\rho_0\rho_1)=0$. 
\end{itemize}

\vspace{0.75em}
\noindent\textbf{Unitary construction using the SWAP test.} 
We utilize the SWAP test to test the closeness of quantum (mixed) states $\rho_0$ and $\rho_1$. Our approach involves a single-qubit quantum register $\sfC$, along with quantum registers $\sfA=(\sfA_0,\sfA_1)$ and $\sfS=(\sfS_0,\sfS_1)$, all initialized to the state $\ket{0}$. Subsequently, we apply state-preparation circuits $Q_i$ on registers $\sfA_i$ and $\sfS_i$ for $i \in \binset$. Then, we perform the SWAP test on registers $\sfC$, $\sfS_0$, and $\sfS_1$, where $\sfC$ serves as the control qubit. Leveraging \Cref{lemma:swap-test} (the SWAP test), we obtain the following unitary (i.e., algorithm $\calA$):
\begin{equation}
    \label{eq:NQP-unitary}
    U\ket{0}_{\sfC}\ket{\bar{0}}_{\sfA,\sfS} = \sqrt{p} \ket{0}_{\sfC} \ket{\psi_0}_{\sfA,\sfS} + \sqrt{1-p} \ket{1}_{\sfC} \ket{\psi_1}_{\sfA,\sfS}, \text{ where } p=\frac{1+\Tr(\rho_0\rho_1)}{2}.
\end{equation}

Next, we introduce another single-qubit register $\sfF$, initialized to zero, leading to:
\begin{equation}
    \label{eq:NQP-unitary-with-flag-qubit}
    \begin{aligned}
    &(H\otimes U) \ket{0}_{\sfF}\ket{0}_{\sfC}\ket{\bar{0}}_{\sfA,\sfS}\\
    =&\sum\nolimits_{k_0\in\binset}\sqrt{\frac{p}{2}} \ket{0}_{\sfF} \ket{k_0}_{\sfC} \ket{\psi_0}_{\sfA,\sfS}
    + \sum\nolimits_{k_1\in\binset}\sqrt{\frac{1-p}{2}} \ket{1}_{\sfF} \ket{k_1}_{\sfC} \ket{\psi_1}_{\sfA,\sfS}\\
     \coloneqq & \sqrt{\frac{p}{2}} \ket{0}_{\sfF} \ket{0}_{\sfC} \ket{\psi_0}_{\sfA,\sfS} + \sqrt{1-\frac{p}{2}} \ket{\perp}_{\sfF,\sfC,\sfA,\sfS}.
    \end{aligned}
\end{equation}

\vspace{0.75em}
\noindent\textbf{Making the error one-sided through exact amplitude amplification.}
Now we devise a one-sided error algorithm $\calA'$ by utilizing $\calA$ as a building block.
Let us consider the Grover operator $G \coloneqq  -(H\otimes U)(I-2\ket{\bar{0}}\bra{\bar{0}}_{\sfF,\sfC,\sfA,\sfS})(H\otimes U^{\dagger})(I-2\Pi_0)$ where $\Pi_0$ is the projector onto the subspace spanned by $\{\ket{0}_{\sfF}\ket{0}_{\sfC}\ket{\phi}_{\sfA,\sfS}\}$ over all $\ket{\phi}$. 
By utilizing the exact amplitude amplification (\Cref{lemma:exact-amplitude-amplification}), we know that $G(H\otimes U)\ket{0}_{\sfF}\ket{0}_{\sfC}\ket{\bar{0}}_{\sfA,\sfS} = \sin(3\theta) \ket{0}_{\sfF}\ket{0}_{\sfC}\ket{\psi_0}_{\sfA,\sfS} + \cos(3\theta) \ket{\perp}_{\sfF,\sfC,\sfA,\sfS}$ where $\theta \in [0,\pi/4]$. 
According to \Cref{eq:NQP-unitary-with-flag-qubit}, $p$ satisfies $\sin^2{\theta} = p/2$. 
Let $x_{\sfF}$ and $x_{\sfC}$ be the measurement outcomes of the registers $\sfF$ and $\sfC$, respectively, after a single iteration of $G$. 
The resulting algorithm $\calA'$ rejects if $x_{\sfF} = x_{\sfC} = 0$; otherwise, it accepts. 
Therefore, the acceptance probability of $\calA'$ is $p_{\rm acc} = 1-\Pr{x_{\sfF} = x_{\sfC} = 0}$ where $\Pr{x_{\sfF} = x_{\sfC} = 0}$ satisfies: 
\begin{equation}
    \label{eq:NQP-acceptance-probability}
    \begin{aligned}
    \Pr{x_{\sfF} = x_{\sfC} = 0} 
    = \sin^2(3\theta)
    = \sin^6{\theta} - 6\cos^2{\theta} \sin^4{\theta} + 9 \cos^4{\theta} \sin^2{\theta}
    = 2p^3 - 6p^2 + \frac{9}{2}p
\end{aligned}
\end{equation}
Finally, we complete the analysis of $\calA'$ as follows: 
\begin{itemize}[topsep=0.33em, itemsep=0.33em, parsep=0.33em]
    \item For \textit{yes} instances, we can plug $\Tr(\rho_0\rho_1) \geq 2^{-n}$ into \Cref{eq:NQP-unitary}, which implies $p \geq \frac{1}{2}+2^{-n-1}$. Noting that $2p^3 - p^2+\frac{9}{2}p \leq 1- (p-\frac{1}{2})^2$ for any $0\leq p \leq 1$, together with \Cref{eq:NQP-acceptance-probability}, we obtain $p_{\rm acc} \geq (p-\frac{1}{2})^2 \geq 2^{-2n-2}$.
    \item For \textit{no} instances, we can set $\Tr(\rho_0\rho_1)=0$ in \Cref{eq:NQP-unitary}, resulting in $p=\frac{1}{2}$ and $\theta = \frac{\pi}{6}$. Following \Cref{eq:NQP-acceptance-probability}, we know that $\calA'$ rejects with certainty, namely $p_{\rm acc}=0$.
\end{itemize}
We thus conclude that $\calA'$ is an \NQP{} algorithm as desired. 
\end{proof}

As mentioned earlier, \Cref{thm:easy-regimes-for-QSZK}\ref{thmitem:QSZK-easy-regimes-NQP} has a classical counterpart, namely $\overline{\SDP}[1,0]$ is in \NP{}. The proof of this folklore result is outlined below: 
We define the collision distance between $p_0$ and $p_1$ as $\Col(p_0,p_1) \coloneqq \sum_{x} p_0(x)p_1(x)$. 
It follows that $\Col(p_0,p_1)=0$ if $\SD(p_0,p_1)=1$. Conversely, when $\SD(p_0,p_1)=0$, we have $\Col(p_0,p_1) \geq 1/|\supp{p_0}\cap\supp{p_1}|$, with equality occurring when $p_0$ and $p_1$ are uniform on $\supp{p_0}\cap\supp{p_1}$.
This observation suffices for establishing the \NP{} containment of $\overline{\SDP}[1,0]$.\footnote{In particular, note that there exists $x\in\supp{p_0}\cup\supp{p_1}$ for $\SD(p_0,p_1)=0$, then the prover could provide the corresponding $w_0$ and $w_1$ as a witness such that $C_0(w_0)=C_1(w_1)=x$. Additionally, such a witness does not exist for \textit{no} instances, i.e.,  $\SD(p_0,p_1)=1$.}

\subsection{\coQSDP{} with some inverse-exponential errors is in \PP{}}
\label{subsec:negl-coQSDP-in-PP}

The crucial insight for comprehending the \PP{} containment of \coQSDP{} with tinily errors is given by the following expression 
\[\frac{1}{2}\HSsquare(\rho_0,\rho_1) = \frac{1}{2} \Tr(\rho_0-\rho_1)^2 = \frac{1}{2} \big( \Tr(\rho^2_0)+\Tr(\rho^2_1) \big) - \Tr(\rho_0\rho_1).\]

It is noteworthy that by employing the SWAP test, one can estimate these three terms: $\Tr(\rho^2_0)$, $\Tr(\rho^2_1)$, and $\Tr(\rho_0\rho_1)$. This estimation enables the development of a hybrid algorithm. 
Subsequently, we proceed to establish \Cref{thm:easy-regimes-for-QSZK}\ref{thmitem:QSZK-easy-regimes-PP}, which can be viewed as the quantum counterpart of~\cite[Theorem 7.1]{BCHTV19}.

\begin{proof}[Proof of \Cref{thm:easy-regimes-for-QSZK}\ref{thmitem:QSZK-easy-regimes-PP}]
    Consider two $n$-qubit quantum states, denoted as $\rho_0$ and $\rho_1$, defined in a finite-dimensional Hilbert space $\mathcal{H}$ according to \Cref{def:QSDP}. 
    We begin with the inequalities between the trace distance and the Hilbert-Schmidt distance~\cite[Equation 6]{CCC19}:
    \begin{equation}
        \label{eq:traceDist-vs-HS}
        \frac{1}{\sqrt{2}} \HS(\rho_0,\rho_1) \leq \td(\rho_0,\rho_1) \leq  \sqrt{\frac{\rank(\rho_0)\rank(\rho_1)}{\rank(\rho_0)+\rank(\rho_1)}} \HS(\rho_0,\rho_1) \leq \frac{\sqrt{\dim\calH}}{\sqrt{2}} \HS(\rho_0,\rho_1). 
    \end{equation}
    
    We present a hybrid classical-quantum algorithm $\calA$ as follows. First, we toss two random coins that the outcomes denoted as $r_1$ and $r_2$. Subsequently, we apply the SWAP test (\Cref{lemma:swap-test}) on the corresponding states in the following manner:
    \begin{itemize}[topsep=0.33em, itemsep=0.33em, parsep=0.33em]
        \item If the first coin lands on heads ($r_1 = 1$), we perform the SWAP test on $\rho_0$ and $\rho_1$. We accept if the final measurement outcome is $0$. 
        \item If the first coin lands on tails ($r_1 = 0$), we perform the SWAP test on two copies of $\rho_{r_2}$. We accept if the final measurement outcome is $1$. 
    \end{itemize}
    
    Let $p^{(o)}_{\SWAP}(\rho_0,\rho_1)$ be the probability of the SWAP test on $\rho_0$ and $\rho_1$ where the final measurement outcome $o$. 
    We then obtain the acceptance probability of our algorithm $\calA$: 
    \begin{equation}
        \label{eq:SWAP-test-HSdistance}
        \begin{aligned}
            \frac{1}{2}p^{(0)}_{\SWAP}(\rho_0,\rho_1) + \frac{1}{2} \sum_{i\in\binset} \frac{p^{(1)}_{\SWAP}(\rho_i,\rho_i)}{2}
            &=\frac{1+\Tr(\rho_0\rho_1)}{4}+ \sum_{i\in\binset} \frac{1-\Tr(\rho^2_i)}{8}\\
            &=\frac{1}{2} - \frac{\HSsquare(\rho_0,\rho_1)}{8}.
        \end{aligned}
    \end{equation}
    
    It suffices to show that algorithm $\calA$ is indeed a \PP{} containment distinguishing \textit{yes} instances from \textit{no} instances within an inverse-exponential gap. 
    Combining \Cref{eq:traceDist-vs-HS,eq:SWAP-test-HSdistance}, we then analyze the acceptance probability:
    \begin{itemize}
        \item For \textit{yes} instances, noting that $\td(\rho_0,\rho_1) \leq 2^{-n/2-1}$, the following holds:        
        \[p^{\sf (Y)}_{\calA}(\rho_0,\rho_1) = \frac{1}{2}-\frac{1}{4}\left(\frac{\HS(\rho_0,\rho_1)}{\sqrt{2}}\right)^2 \geq \frac{1}{2}- \frac{\td^2(\rho_0,\rho_1)}{4} \geq \frac{1}{2}-2^{-n-4}.\]
        \item For \textit{no} instances, noticing that $\td(\rho_0,\rho_1) \geq 1-2^{-n/2-1}$, it holds that:
        \[p^{\sf (N)}_{\calA}(\rho_0,\rho_1) = \frac{1}{2}-\frac{1}{4}\rbra*{ \frac{\HS(\rho_0,\rho_1)}{\sqrt{2}} }^2 \leq \frac{1}{2}-\frac{1}{4}\cdot\frac{\td(\rho_0,\rho_1)^2}{\dim\calH} \leq \frac{1}{2}-2^{-n-2}\cdot \rbra*{ 1-2^{-\frac{n}{2}-1} }^2.\]
        % =1-2^{-2n-3}-2^{-n-1}+2^{-3/2 n-1}
    \end{itemize}
        
    Since $\PreciseBQP \subseteq \PP$, such as~\cite[Lemma 3.3]{GSSSY22}, we complete the proof by showing that the gap $p^{\sf (Y)}_{\calA}(\rho_0,\rho_1) - p^{\sf (N)}_{\calA}(\rho_0,\rho_1)$ is exponentially small as desired:
    \[p^{\sf (Y)}_{\calA}(\rho_0,\rho_1) - p^{\sf (N)}_{\calA}(\rho_0,\rho_1)= 2^{-2n-4} + 2^{-n-2}\cdot \rbra*{\frac{3}{4} - 2^{-n/2}} \geq 2^{-2n-4},\]
    % =2^{-2n-3} + 2^{-n}(1/2 - 1/8 -1/2 2^{-n/2})
    where the last inequality holds for $n\geq 1$. 
\end{proof}

%%%%%%%%%%%%%%%%%%%%%%%%%%%%%%%%%%%%%%%%%%%%%%%%%%%%%

\section*{Acknowledgments}
\noindent
An earlier version of this work was included in the author's PhD thesis~\cite{Liu25}. 
The author expresses gratitude to François Le Gall for providing valuable suggestions to improve the presentation and engaging in insightful discussions. The author also thanks Qisheng Wang for proposing the use of the SWAP test in \Cref{thm:easy-regimes-for-QSZK}\ref{thmitem:QSZK-easy-regimes-NQP}. 
Moreover, the author thanks anonymous reviewers for enlightening comments on the polarization lemma, pointing out an error in \Cref{thm:easy-regimes-for-QSZK}\ref{thmitem:QSZK-easy-regimes-NQP} in an earlier version, mentioning previous uses of the quantum Jensen-Shannon divergence in quantum communication complexity, and useful suggestions for improving the presentation. 
The author was supported by JSPS KAKENHI Grants No.~JP20H04139, as well as JST, the establishment of University fellowships towards the creation of science technology innovation, Grant No.~JPMJFS2125. 
Circuit diagrams were drawn by the Quantikz package~\cite{Kay18}. 
%%%%%%%%%%%%%%%%%%%%%%%%%%%%%%%%%%%%%%%%%%%%%%%%%%%%%

% Reference
\bibliographystyle{alpha}
\bibliography{QSZK_state_testing}

\end{document}